\documentclass[11pt,a4paper]{article}
\pdfoutput=1
\usepackage{amssymb,amsmath,amsfonts, mathtools, mathrsfs}
\usepackage[utf8]{inputenc} % per usare le lettere accentate come à,è etc...
\usepackage[dvipsnames]{xcolor}
\usepackage{amsfonts}
\usepackage{amsmath}
\newif\ifnatbibsort\natbibsorttrue

\relax

\ifnatbibsort\RequirePackage[numbers,sort&compress]{natbib}\else\RequirePackage[numbers,compress]{natbib}\fi
\RequirePackage[colorlinks=true
,urlcolor=blue
,anchorcolor=blue
,citecolor=blue
,filecolor=blue
,linkcolor=blue
,menucolor=blue
,pagecolor=blue
,linktocpage=true
,pdfproducer=medialab
,pdfa=true
]{hyperref}

\usepackage{graphicx}
\usepackage{caption}
\usepackage{tensor}
\usepackage{subfigure}
\usepackage{enumerate}
\usepackage{dsfont}
\usepackage{bbold}
\usepackage{verbatim}  
\usepackage{amsfonts,euscript,amssymb,stmaryrd,braket}
\usepackage{tikz}
\usetikzlibrary{arrows,decorations.markings,patterns}
\usepackage{slashed}

\def\clock{{\count0=\time
		\divide\count0 60
		\ifnum\count0<10 0\fi\the\count0
		\multiply\count0 -60 \advance\count0 \time
		:\ifnum\count0<10 0\fi \the\count0
}}
\newcommand{\timestamp}{{\small\vbox{\hbox{\tt\jobname.tex}
			\hbox{\the\day/\the\month/\the\year, \clock}}}}

\newcommand{\Ecal}{{\cal E}}
\newcommand{\Ucal}{{\cal U}}

\newcommand{\Ccal}{{\cal C}}

\newcommand{\bea}{\begin{eqnarray}}
\newcommand{\eea}{\end{eqnarray}}

\newcommand{\bs}[1]{\boldsymbol{#1}}
\newcommand{\xv}{{\bs{x}}}
\newcommand{\yv}{{\bs{y}}}

\DeclareMathOperator{\diag}{diag}

\newcommand{\be}{\begin{equation}}
\newcommand{\ee}{\end{equation}}

\newcommand{\Mtilde}{\widetilde{M}}

\makeatletter
\let\old@startsection=\@startsection
\let\oldl@section=\l@section
\renewcommand{\@startsection}[6]{\old@startsection{#1}{#2}{#3}{#4}{#5}{#6\mathversion{bold}}}
\renewcommand{\l@section}[2]{\oldl@section{\mathversion{bold}#1}{#2}}
\makeatother

\numberwithin{equation}{section}
\def \be {\begin{equation}}
\def \ee {\end{equation}}
\def \ba {\begin{array}}
\def \ea {\end{array}}
\def \bea{\begin{eqnarray}}
\def \eea{\end{eqnarray}}

\def \Tr {{\textrm{Tr}}}

\def \diag {{\textrm{diag}}}

\def \and {{~\textrm{and}~}}

\newcommand{\half}{\frac{1}{2}}

\usepackage{color}

\newtheorem{thm}{Theorem}

\begin{document}
	\renewcommand{\thefootnote}{\arabic{footnote}}

	\overfullrule=0pt
	\parskip=2pt
	\parindent=12pt
	\headheight=0in \headsep=0in \topmargin=0in \oddsidemargin=0in

	\vspace{ -3cm} \thispagestyle{empty} \vspace{-1cm}
	\begin{flushright} 
		\footnotesize
		\textcolor{red}{\phantom{print-report}}
	\end{flushright}

%\begin{center}
%	\vspace{.0cm}
%	
%	{\Large\bf \mathversion{bold}
%	Capacity of entanglement and monotones
%	}
%	\\
%	\vspace{.25cm}
%	\noindent
%	{\Large\bf \mathversion{bold}
%	in free theories}
	
\begin{center}
	\vspace{.0cm}
	
	{\Large\bf \mathversion{bold}
	Sequences of resource monotones \\
	 from modular Hamiltonian polynomials}

	\vspace{0.8cm} {
		Ra\'ul Arias$^{\,a,b}$,
		Jan de Boer$^{\,c}$,
		Giuseppe Di Giulio$^{\,d}$,
		\\
		Esko Keski-Vakkuri$^{\,e,f,g}$
		%\footnote[1]{mintchev@df.unipi.it}
		and Erik Tonni$^{\,h}$
		%\footnote[2]{erik.tonni@sissa.it}
	}
	\vskip  0.5cm
	
	\small
	{\em
		$^{a}\,$Instituto de F\'isica de La Plata, CONICET\\
Diagonal 113 e/63 y 64, CC67, 1900 La Plata, Argentina
		\vskip 0.05cm
		$^{b}\,$Departamento de F\'isica, Universidad Nacional de La Plata,\\ Calle 49 y 115 s/n, CC67, 1900 La Plata, Argentina
		\vskip 0.05cm
		$^{c}\,$Institute for Theoretical Physics and Delta Institute for Theoretical Physics, University of Amsterdam, \\
		PO Box 94485, 1090 GL Amsterdam, The Netherlands
		\vskip 0.05cm
		$^{d}\,$Institute for Theoretical Physics and Astrophysics and W\"urzburg-Dresden Cluster of Excellence ct.qmat, Julius-Maximilians-Universit\"at W\"urzburg, Am Hubland, 97074 W\"{u}rzburg, Germany

		\vskip 0.05cm
		$^{e}\,$Department of Physics, University of Helsinki \\
		PO Box 64, FIN-00014 University of Helsinki, Finland
                     \vskip 0.005cm 
                      $^{f}\,$Helsinki Institute of Physics \\
                      PO Box 64, FIN-00014 University of Helsinki, Finland
		\vskip 0.05cm
                      $^{g}\,$InstituteQ - the Finnish Quantum Institute, University of Helsinki, Finland
		\vskip 0.05cm
		$^{h}\,$SISSA and INFN Sezione di Trieste, via Bonomea 265, 34136, Trieste, Italy 
	}
	\normalsize

\end{center}

\vspace{-0.3cm}
\begin{abstract} 
We introduce two infinite
sequences of entanglement monotones, which are constructed from expectation values of polynomials in the modular Hamiltonian. 
%can be computed from the R\'enyi entropies, involving along with von Neumann entropy also higher cumulants of
%the modular Hamiltonian. 
These monotones yield  infinite sequences of inequalities  that must be satisfied in majorizing state transitions. 
We demonstrate this for information erasure, deriving an infinite sequence of "Landauer inequalities" for the work cost, 
bounded by linear combinations of expectation values of powers of the modular Hamiltonian. These inequalities give improved lower bounds for the 
work cost in finite dimensional systems, and depend on more details of the erased state than just on its entropy and variance of modular Hamiltonian.   
Similarly one can derive lower bounds for marginal entropy production for a system
coupled to an environment.  These infinite sequences of entanglement monotones also give rise to relative quantifiers that
are monotonic in more general processes, namely those involving so-called $\sigma$-majorization with respect to a fixed point full rank state $\sigma$; such quantifiers are called resource monotones. As an application to thermodynamics, one can use them to derive finite-dimension corrections
to the Clausius inequality.
Finally, in order to gain some intuition for what (if anything) plays the role of majorization in field theory, 
we compare pairs of states in discretized theories at criticality and study how majorization depends on the size of the bipartition with respect to the size of the entire chain.

\end{abstract}

%\vspace{1cm}
%
%\begin{center}
%{\bf February 10, 2022}
%\end{center}

\newpage

%%%%%%%%%%%%%%%%%%%%%%%%%%%%%%%%%%%%%
\tableofcontents
%%%%%%%%%%%%%%%%%%%%%%%%%%%%%%%%%%%%%

\newpage

%%%%%%%%%%%%%%%%%%%%%%%%%%%%%%%%%%%%%%%%%%%%%%%%%%%%%

\section{Introduction}

The von Neumann entropy $S(\rho )=-\Tr (\rho \ln  \rho )$ of a quantum state $\rho$ is perhaps the best known quantifier of quantum information. The same can be said about its application to a reduced density 
matrix $\rho_A=\Tr_{B}(\rho)$ of a global state $\rho$ (with $\Tr_A  (\rho_A) $ =1) for a bipartite entangled quantum system $A\cup B$, the entanglement entropy $S_A(\rho_A)$.  Also widely known are the 
\noindent
R\'enyi entropies
\be
\label{defRenyi}
S^{(n)}=\frac{1}{1-n}\ln\textrm{Tr}\rho^n\,,
\ee
and the limit $n\to 1$ relating them to the Von Neumann entropy
\be
\label{defentropy}
S \equiv
\lim_{n \to 1} S^{(n)}
=
-\,\partial_n \big( \textrm{Tr}\rho^n \big)\big|_{n=1}
=
-\,\partial_n \big( \ln \textrm{Tr}\rho^n \big)\big|_{n=1}
=
-\,\textrm{Tr}(\rho\ln\rho ) \ .
\ee
The definition of $S(\rho)$ allows to interpret it as the expectation value of the Hermitean operator $K=-\ln \rho$, well defined since $\rho$  is hermitian and positive definite.  The standard normalization is\footnote{With this normalization, 
{\em e.g.} for a thermal state $e^{-\beta H}/Z$ we have $K=\beta H + \ln Z$.} $\Tr ( \rho ) = \Tr (e^{-K})=1$. Following the convention from \cite{Haagbook}, we call $K=-\ln \rho$ the {\em modular Hamiltonian}\footnote{In quantum information theory $-\ln \rho$ is sometimes called {\em surprisal}, see {\em e.g.} \cite{boes2020}. This name is inherited from its classical equivalent $-\ln p_i$ in the context of a classical discrete
probability distribution ${p_i},\ \sum_i p_i=1$: an outcome $i$ with a very low probability is associated with a large surprise -- consider for example
winning the big prize in a lottery. In a quantum version, using a diagonal basis $\rho = \sum_i p_i |i\rangle \langle i|$, if a measurement in the eigenbasis gives as a result an eigenvalue $i$
with very low probability, the surprise is large. Surprisal is also naturally associated with information: obtaining a rare measurement outcome can be interpreted to reveal a large amount
of information about the possible eigenvalues. Indeed another common name for surprisal is {\em information content.} The von Neumann entropy is thus the expected or average value of surprisal or revealed information, associated with measurement 
in the eigenbasis.}.
In the context of a bipartite system and entanglement, the Hermitean operator $K_A\equiv -\ln\rho_A $, derived from a reduced density matrix $\rho_A$, is also
often called
the {\em entanglement Hamiltonian}. The entanglement entropy $S_A$ is the average (or the first moment) of the entanglement Hamiltonian, namely $S_A=\langle K_A \rangle$, where the mean value is evaluated through $\rho_A$.

In addition to the first moment of modular Hamiltonian, it is natural to explore its higher moments or cumulants as well. The second cumulant, the variance of modular Hamiltonian, is a much
less-known quantity, and therefore even its name varies in the literature. It is also known as entropy variance, varentropy, and in the context of entanglement in many-body physics and quantum field theory, as {\em capacity of entanglement} $C_A(\rho_A )$. In the latter context, it was introduced in \cite{Yao:2010woi} and \cite{schliemann}, first with a definition modeled after that
of a heat capacity, and proposed to detect different phases in topological matter. Since heat capacity is related to the variance of thermodynamical entropy, it was realized that
capacity of entanglement is equal to the variance of the modular Hamiltonian, and can also be derived from the R\'enyi entropies  \cite{Perlmutter:2013gua, deBoer:2018mzv, Nakaguchi:2016zqi}
\be
\label{defcapacity}
C_A\,=\,
\partial^2_n \big( \ln \textrm{Tr}\rho_A^n \big)\big|_{n=1}
=\,
\partial_n^2 \big( \textrm{Tr}\rho_A^n \big)\big|_{n=1}
-
\big[ \partial_n \big( \textrm{Tr}\rho_A^n \big)\big]^2 \big|_{n=1}
=
\,\langle K_A^2 \rangle-\langle K_A \rangle^2 \, .
\ee
%
%%%%%%%%%%%%%% from esko Feb 21 %%%%%%%%%%%%%%%
One of the reasons the variance of the modular Hamiltonian or capacity of entanglement is less-known, is that as a quantifier it is not known to satisfy many interesting properties, unlike the von Neumann 
entropy does. Some of its uses in quantum information theory, which we are aware of, are in a finite-size correction to the Landauer inequality or more generally in bounding the increase in entropy
in state transitions between majorizing states \cite{Reeb_2014, boes2020}, in the analysis of catalytic state transformations \cite{boes2020}, and in state interconvertibility in
finite systems \cite{Chubb2018}.

This work is primarily motivated by \cite{boes2020}, which considered the capacity of entanglement (there called variance of surprisal) and more generally the relative variance %\textcolor{blue}{Should we call this quantity {\it relative variance}?}
\be
\label{def_relative variance}
  C(\rho || \sigma ) = \Tr [\rho (\ln \rho - \ln \sigma)^2 ] - S(\rho ||\sigma )^2 \ ,
\ee
where $S(\rho ||\sigma )\equiv \Tr [\rho (\ln \rho - \ln \sigma) ]$ is the relative entropy. When we consider $\sigma$ in (\ref{def_relative variance}) to be the maximally mixed state $C(\rho || \sigma )$ reduces to $C(\rho)$. While both the relative variance and the relative entropy have applications in the independent and identically distributed setting involving many copies of a system or operations, the authors of \cite{boes2020} explored their role in 
a single-shot setting where an operation or protocol is executed only once in one system, proving many new results for the variance and relative variance. Many properties were
based on  a new quantifier\footnote{Ref. \cite{boes2020} uses the convention where definitions involve the binary logarithm $\log_2(x)$ instead of $\ln (x)$. 
%\bea 
%&& S(\rho) = -\Tr [\rho \log_2 (\rho )] \ , \nonumber \\
%&& C(\rho) = \Tr [\rho (- \log_2 (\rho) )^2]-S^2 (\rho) \\
%&& M(\rho) = C(\rho) + \big( S(\rho) +\frac{1}{\ln 2} \big)^2 \ . \nonumber 
%\eea
We prefer to follow the physics convention and use the natural logarithm in definitions. Since $\ln(x) = (\ln 2)\log_2 (x)  $, denoting below quantities defined with $\log_2$ by tildes ({\em e.g.} $\tilde{S}=-\Tr [\rho \log_2 \rho]$), we have
$S = (\ln 2)\tilde{S},\ C= (\ln 2)^2\tilde{C} $ and $M=(\ln 2)^2\tilde{M}$ with 
\be
\tilde{M}(\rho) = \tilde{C}(\rho) + \bigg(\tilde{S}(\rho) + \frac{1}{\ln 2}\bigg)^2 \ ,
\ee 
as given in \cite{boes2020}.}
\be\label{Mdefinition}
 M(\rho ) = C(\rho) + \left[S(\rho) + 1 \right]^2  \ ,
\ee   
which was shown to be Schur concave, and its relative version which was shown to be a resource monotone. The quantities $S,C,M$ were shown to be connected by an inequality\footnote{This inequality has the same form with the $\log_2$-conventions:
$$
\tilde S(\rho ) -\tilde S(\sigma ) \geq \frac{\tilde C(\sigma )-\tilde C(\rho)}{2\sqrt{\tilde M(\sigma)}}\, .
$$
}
\be\label{Landauer-ineq}
S(\rho ) -S(\sigma ) \geq \frac{C(\sigma )-C(\rho)}{2\sqrt{M(\sigma)}}\, ,
\ee
when the majorization order $\sigma \succ \rho$ holds for two states in a system with a finite dimensional Hilbert space. As an application, \cite{boes2020} considered {\em e.g.} information erasure, deriving
a new lower bound for the associated work cost that involves both entropy and variance. Related results were proven for the relative quantifiers. 
In the end they  posed a question whether it is possible to extend this construction of $M$ to a sequence of Schur concave quantifiers that would have similar properties and involve higher cumulants than $C$, perhaps also likewise
for the relative quantifiers. 

In state conversions involving a majorization order between the initial state $\rho$ and the final state $\sigma$, the whole spectrum of eigenvalues is affected, and the majorization order itself can be defined by a sequence of inequalities. 
Properties of the spectra can be characterized by various quantities, such as R\'enyi entropies, or moments and cumulants of
$-\ln \rho$ and $-\ln \sigma$. Hence it is natural to expect that majorization order may imply a sequence of inequalities involving changes in cumulants beyond the first two. To derive such inequalities,  we construct two sequences of entanglement monotones, that can be expanded as combinations of cumulants.
We first generalize the construction (\ref{Mdefinition}) and define the {\em moments of the shifted modular Hamiltonian} 
\be
M^{(n)}(\rho;b_n) = \Tr [\rho (-\ln \rho+ b_n)^n] -b^n_n\,,
\ee
for $n\geq 1$ (with $M^{(2)}(\rho ;b_n = 1)=M (\rho)-1$). Explicit formulas for their expansions by higher cumulants are given in Sec.\,\ref{sec:resource_monotones}. For the parameter range $b_n\geq n-1$ all of them are concave (see Sec.\,\ref{subsec:concavity} for the relevant definitions), hence, from the Vidal's theorem \cite{Vidal}, they are pure state entanglement monotones (where $\rho$ and $\sigma$ are the reduced density matrices of global pure states $\ket{\psi}$ and $\ket{\phi}$ respectively), thus yielding inequalities
\be\label{Minequality}
M^{(n)}(\rho ;b_n) \geq M^{(n)}(\sigma ;b_n)\,, 
\ee
 in local operations assisted with classical communication (LOCC) and other majorizing state transformations\footnote{For LOCC transformations of mixed states, we need to apply the convex roof extension to $M^{(n)}$ (see (\ref{convexroofextension})).} with $\sigma \succ \rho$.
For example, at second order $n=2$ with $b_n=1$ we obtain the inequality
\be
\label{inequality_new n2}
S(\rho)-S(\sigma) \geq \frac{C(\sigma )-C(\rho )}{S(\rho )+S(\sigma )+2}\,,
\ee
which is slightly sharper than the inequality (\ref{Landauer-ineq}).
%, which in turn is obtained from (\ref{Minequality}) with $n=2,\, b_n=1$.
We also show how to  calculate the moments $M^{(n)}$ from  R\'enyi entropies, by using the latter as a generating function. 
The R\'enyi entropies $S^{(n)}$ are not concave (for index value $n > 1$), hence our observation gives a way to repackage their information to an infinite sequence of concave quantifiers, that define entanglement monotones. 

Next we identify a basis for monotones which are polynomial in moments of $-\log \rho$ which allows to construct another infinite sequence that we call {\it extremal polynomial monotones} $P^{(n)}_E$ (see Sec.\,\ref{subsec:ExtremalMonotones}). All  $-P^{(n)}_E$ are also concave, hence define monotones, and moreover any concave polynomial can we written as a linear combination of extremal polynomial monotones with non-negative coefficients. 
We therefore believe that
they provide the tightest inequalities of this type in majorizing state transformations. For example, given two majorizing states $\rho\succ\sigma$, the third monotone $P^{(3)}_E$ yields the inequality
\be\label{Ineq3compact_intro}
\Delta M_3 \geq 3 \Delta M_2 + \frac{3}{4} \frac{(\Delta M_2)^2}{\Delta M_1}\,,
\ee
and the fourth order $P^{(4)}_E$ the inequality
\be\label{P4intro}
\Delta M_4 \geq 8 \Delta M_3 -6\Delta M_2 + \frac{8}{9}\frac{(\Delta M_3-3\Delta M_2)^2}{\Delta M_2} \, ,
\ee
where $\Delta M_n \equiv  M^{(n)} (\sigma ; n-1) -  M^{(n)} (\rho ; n-1)$. Notice that (\ref{Ineq3compact_intro}) and (\ref{P4intro}) are stronger than the inequalities   $\Delta M_3\geq 0$ and $\Delta M_4 \geq 0$ obtained from (\ref{Minequality}). 

As an application of these new inequalities  we first consider information erasure: we obtain infinite sequences of "Landauer inequalities" for the work cost, bounded by arbitrarily high cumulants of the modular Hamiltonian of the initial state $-\ln\rho$ to be erased, extending the previous result of \cite{boes2020}, that involves only the variance. We also derive a slightly sharper inequality for marginal entropy production applying a unital
quantum channel to a system and environment, and outline steps for deriving an infinite sequence of inequalities. 

For relative quantifiers, we  first generalize a  theorem proven in \cite{boes2020}, to show how one can construct an infinite class of resource monotones, relative quantifiers based on a concave quantifier $E(\rho)=\Tr[\rho F(\ln \rho)]$. We then
apply this construction to the monotones $M^{(n)}$ and $-P^{(n)}_E$, obtaining infinite sequences of resource monotones that involve cumulants of $\ln \rho -\ln \sigma$. An important restriction is that the results only apply in the "classical" case where $\rho$ and $\sigma$ commute. The sequences imply inequalities for relative entropy production bounded
by changes in the relative cumulants; we consider two examples more explicitely. In particular, as an application to quantum thermodynamics, we derive a finite-size correction to the Clausius inequality:
\be\label{Clausiusintro}
 S (\gamma_\beta )-  S (\rho) 
\geq \frac{1}{k_BT}[\langle H\rangle_{\gamma_\beta} - \langle H \rangle_\rho] +\frac{C(\rho ||\gamma_\beta )}{2+ 2\beta (E_{\rm max} -F(\beta))}\ ,
\ee
where $\rho$ is a non-equilibrium state commuting with the equilibrium thermal state $\gamma_\beta = \exp (-\beta H)/Z(\beta)$,  $E_{\rm max}$ is the maximum energy eigenvalue of the finite-dimensional system, and $F(\beta)=-\beta^{-1}\ln Z(\beta)$ is the Helmholtz free energy.
Here we refer to finite-size corrections as corrections which vanish when the dimension of the Hilbert space of the system is infinite. They should not be confused with the corrections related to the occurrence of a finite volume.

The above results apply to majorization in finite dimensional systems. In quantum field theory, less is known about majorization. There have been studies
investigating ground state entanglement and the behaviour of majorization in the reduced density matrix in a subsystem under renormalization group flow and scaling transformations  \cite{Vidal:2002rm,latorre2003ground,Latorre:2004pk,Orus:2005jq,Riera:2006vj,Calabrese-Lefevre,Swingle:2013zla,Mintchev:2022xqh}. In this work our interest is in the possibility of majorization between a pair of states in a quantum field theory. We explore this question to gain some insight by considering pairs of states in 1+1 dimensional
conformal field theories (CFT),  in particular free theories such as a compact boson and a Dirac fermion. As a pair of states we take the ground state and an excited state, then discretize the theory and map it to a fermionic chain, where
we find the corresponding pair of states (yielding back to the CFT states in a continuum limit).  Furthermore, we take the theory to live on a circle with periodic boundary conditions. The discrete fermionic chain then has a finite dimensional Hilbert space, so we can use
the standard definition of majorization. We bipartite the theory into a line segment and its complement and ask if a majorization order exists between the pair of pure states. The majorization condition involves reduced states, 
which depends on the bipartition, and thus on the relative  size of the subsystem. While it is laborious to directly verify the majorization conditions, it is simpler to show them to be violated by comparing entanglement monotones or Schur concave
quantifiers for the pair of states and test if the majorization-implied inequality is falsified for any monotone. In the final part of this work we perform such comparisons, between the ground state and an excited state in a CFT and the 
corresponding pair in the periodic fermionic chain. We consider the entanglement entropy, the R\'enyi entropies $S^{(2)},\, S^{(3)}$ and the monotone $M^{(2)}$, and compare which quantity gives the most stringent bound ruling
out majorization in the largest range of bipartition.

This paper is organized as follows. In Sec.\,\ref{subsec:concavity} we first review some of the relevant basic concepts of quantum information theory.  In Sec.\,\ref{sec:resource_monotones} we introduce two sequences of entanglement monotones: the moments of shifted modular Hamiltonian and the extremal polynomial monotones. We then generalize a theorem in \cite{boes2020}, allowing us to construct infinite sequences of relative quantifiers which are also resource monotones. As an application,
we derive a finite-size correction to the Clausius inequality. In Sec.\,\ref{sec:applications} as an application of the new entanglement monotones, we consider Landauer erasure, and derive infinite sequences of inequalities for the work cost of the erasure process, involving arbitrarily high cumulants of the  state to be erased. We also discuss bounds on marginal entropy production in a system coupled to an environment. We then start an initial exploration of state majorization in CFTs in Sec.\,\ref{CFT}.
We compute the quantities $S,C, M$ in some simple free CFTs for the ground state and excited states, and in corresponding discretized fermionic chains, and we examine the monotonicity of some Schur concave quantifiers as proxies for state majorization, as outlined in Sec.\,\ref{sec:applications}.  We conclude with a discussion and a description of various open problems.

%\newpage

\section{Some concepts of quantum information theory}
\label{subsec:concavity}
%\label{QuantumInfo}

%%%%%%%%%%%%%%%%%%%%%%%%%%%%%%%%%%%%%%%%%%%%%%%%%%
%%%%%%%% New subsection (former Appendix A) %%%%%%%%%%%%%%%%%%%%%%%%%%
%%%%%%%% from esko Feb 21 %%%%%%%%%%%%%%%%%%%%%%%%%%%%%%%%%%%

For the benefit of readers who are less familiar with some of the relevant concepts of quantum information theory, we briefly review some relevant background material. 

In this work our focus is on bipartite systems $A\cup B$, where the Hilbert space is decomposed as $\mathcal{H}_{A\cup B}= \mathcal{H}_A\otimes\mathcal{H}_B$. For a pair of quantum states described by the density matrices $\rho$ and $\sigma$, we
first review the important concept of majorization (partial) order. Consider a pair of vectors $\boldsymbol{\lambda},\boldsymbol{\kappa}\in \mathbb{R}^d$ and assume them to be ordered so that the components satisfy $\lambda_1\geq \lambda_2 \geq \cdots \geq \lambda_d$ and likewise for $\boldsymbol{\kappa}$.  
We say that $\boldsymbol{\lambda}$ majorizes $\boldsymbol{\kappa}$, and denote $\boldsymbol{\lambda}\succ \boldsymbol{\kappa}$, when \cite{GeomQuantumStates_book}
\bea
\label{majorisation}
\sum_{k=1}^m \lambda_k &\geq& \sum_{k=1}^m \kappa_k\,,
\,\,\qquad\,\,
\forall \ m=1,2,\dots,  d\ .
%\\
%\sum_{k=1}^N \lambda_k &=& \sum_{k=1}^N \kappa_k
\eea
Majorization defines a partial order in $\mathbb{R}^d$, and the definition easily extends to the case $d \to \infty$, where we have a countably infinite number of inequalities to satisfy.
We then define majorization between two density matrices $\rho_1$ and $\rho_2$: 
$\rho_1\succ \rho_2$ when $\boldsymbol{\lambda}_1\succ \boldsymbol{\lambda}_2$,
where $\boldsymbol{\lambda}_j$ is the ordered vector of eigenvalues of $\rho_j$. Note that the inequalities (\ref{majorisation}) become trivial for any pair of pure states. However, in a bipartite system, 
and when this partition is kept fixed, one can define a non-trivial majorization partial order for pure states. Consider a pair of pure states $|\psi \rangle ,|\phi \rangle \in \mathcal{H}_{AB}$, and define majorization following
that of the reduced density matrices,
\be
    |\psi \rangle \succ |\phi \rangle
    \qquad \Leftrightarrow 
    \qquad\Tr_B (|\psi \rangle \langle \psi |) \succ \Tr_B (|\phi \rangle \langle \phi |)  \ .
\ee
Note that it does not matter whether the partial trace is taken over $B$ or $A$ because the resulting reduced density matrices in the two cases have the same eigenvalues, from the Schmidt decomposition. However, we emphasize that the definition depends on the choice of
the bipartition $A\cup B$, and it would be more accurate to denote it by $\ket{\psi} \succ_{A B} \ket{\phi}$:  an alternative bipartition $A'\cup B'$ in general leads to a partial order $\succ_{A'B'}$ among bipartite pure states which is not equivalent with $\succ_{AB}$.

We will be interested in quantities that are monotonic under majorization. First, a function $g$ mapping a density matrix to a real number is said to be {\it Schur concave}, when, for any pair of density matrices, we have
\be
\rho \succ \sigma 
\qquad\Rightarrow
\qquad
 g(\rho )\leqslant g(\sigma) \ . 
\ee
Conversely, $g$ is Schur convex if $-g$ is Schur concave \cite{GeomQuantumStates_book}.
A stronger property is  %(operator) 
concavity, which is a crucial property we employ to construct entanglement monotones.
We say that $g$ is  {\it concave}, when, for any $0\leqslant p\leqslant1$ and for any pair of operators $\rho$ and $\sigma$, we have
\be
\label{def concavity}
g\big(p\rho +(1-p)\sigma )\geq p g(\rho )+(1-p)g(\sigma )\,.
\ee
Conversely, a function $g$ is called convex if $-g$ is concave. A way to construct concave quantities is to begin with a function $f: [0,1]\rightarrow \mathbb{R}$. 
By applying the function to an operator $\rho$, we obtain another operator $f(\rho )$. Assume that all the involved operators are diagonalised by unitaries $\rho =U\textrm{diag}(\lambda_i)U^\dagger$,
we have that $f(\rho)=U\textrm{diag}(f(\lambda_i))U^\dagger$, where we assume that $f$ is well defined for all $\lambda_i$.  We can now define a function $g$ through $f$ as
\be 
\label{def Tr funct}
g(\rho )=\textrm{Tr}\big[f(\rho )\big]\,.
\ee
 It has been proven that if $f$ is concave (convex) as a single real variable function, then $g$ defined in (\ref{def Tr funct}) is concave (convex) according to the definition (\ref{def concavity}) \cite{GeomQuantumStates_book}.  For example, when $f(x)=-x\ln x$ and $\rho $ is a density matrix, 
 the function defined in (\ref{def Tr funct}) is the von Neumann entropy.
Since $f(x)=-x\ln x$ is concave in $[0,1]$, the von Neumann entropy is concave.  

As we said, concavity is a stronger property than Schur concavity. Any function that is symmetric in its arguments (such as $g$ defined in (\ref{def Tr funct}) by the trace, it  is symmetric {\em i.e.} invariant with respect to permutations of $\lambda_i$) and concave is also Schur concave, while the opposite does not hold \cite{GeomQuantumStates_book}. For example the R\'enyi entropies $S^{(\alpha)}=\frac{1}{1-\alpha} \ln \Tr (\rho^\alpha)$ are Schur concave when $\alpha >0$, but concave only when $0< \alpha \leq 1$. Also the
so-called   min and max entropies are Schur concave but not concave. 
On the other hand, 
 the von Neumann and Tsallis\footnote{Given a density matrix $\rho $ and $q>1$ the Tsallis entropy is 
 $S_{q}=\frac{1}{1-q}\big[\textrm{Tr}\rho ^q-1\big]$. In the limit $q\to 1$ the Tsallis entropy gives the von Neumann entropy.} entropies are 
both concave and Schur concave (see \cite{Hu2006} for a demonstration in the context of unified entropies).  

A general way to define quantum operations $\Ecal$ mapping an input state $\rho$ to an output state $\Ecal (\rho)$ is by the operator-sum representation\footnote{Another equivalent way is by the Stinespring dilation theorem, introducing
an environment system $\mathcal{H}_\mathrm{E}$ in a reference state $\rho_\mathrm{E}$ and then represent $\Ecal (\rho )=\Tr_\mathrm{E} [U \rho \otimes \rho_\mathrm{E} U^\dagger ]$ where $U$ can be chosen to be a unitary operator acting in the composite system $\mathcal{H}\oplus \mathcal{H}_\mathrm{E}$
and the partial trace is taken over $\mathcal{H}_\mathrm{E}$. }
\be
  \Ecal (\rho) = \sum_i K_i \rho K^\dagger_i \,,
\ee
with a collection of Kraus operators $\{ K_i \}$ that satisfy the condition 
\be
 \sum_i K^\dagger_i K_i \leq \mathbb{1}\,,
\ee
where, for a Hermitian $A$, $A\leq \mathbb{1}$ means that $\mathbb{1}-A$ has only non-negative eigenvalues. If the stronger condition $\sum_i K^\dagger_i K_i = \mathbb{1}$  applies, then $\Ecal$ is trace-preserving and it is called a quantum channel.
Every quantum channel has a fixed point $\sigma_* = \Ecal (\sigma_*)$ \cite{NielsenChuangBook}. If the fixed point is the unit matrix, $\Ecal (\mathbb{1})=\mathbb{1}$, the operation $\Ecal$ is called unital channel. In this case the Kraus operators satisfy the additional condition
\be
\sum_i K_iK^\dagger_i = \mathbb{1} \ .
\ee
For us, an important feature of unital channels is that by Uhlmann's theorem \cite{Uhlmann-thm, watrous_2018} they imply majorization between the input and output states,
\be\label{Uhlmann-theorem}
     \rho \succ \Ecal (\rho) \ .
\ee
A simple proof (in English) of Uhlmann's theorem, based on the Hardy-Littlewood-P\'olya theorem of majorization (which establishes that $\boldsymbol{\lambda} \succ \boldsymbol{\mu}$ iff there exists a bistochastic matrix $T$ such that  $\boldsymbol{\mu} = T\boldsymbol{\lambda}$),  can be found in \cite{KriSpe} (in Appendix B therein). The converse is also true in the sense that if $\rho \succ \sigma$ there exists a unital channel with $\sigma= \Ecal (\rho)$.

Another well-known class of operations are the LOCC. We can think of a LOCC as a process where quantum operations are performed by the two parties $A$ and $B$ separately, while classical communication allows the two parties to correlate their action. We emphasize again that for this process one must first decide on a bipartition, and then keep it fixed.  Mathematically, LOCC operations can be represented as separable operations  \cite{Plenio:2007zz}
\be\label{separable}
  \rho \mapsto \Lambda (\rho) = \sum_i p_i\, A_i\otimes B_i\, \rho\, A^\dagger_i\otimes B^\dagger_i\,,
\ee
where $A_i$ and $B_i$ are operators acting on the local subsystems $A$ and $B$ respectively. Note however that it is notoriously difficult to characterize the set of operations which can be a achieved through LOCC and that the class of separable operations (\ref{separable}) is strictly larger than LOCC. 
%\textcolor{brown}{Since entanglement cannot be created, but only decreased by LOCC, they can be used to define entanglement.}
The LOCC operations can be used to define entanglement; indeed entanglement cannot
be created but only decreased by these operations. Moreover, separable states, which are states $\rho$ of the form
\be\label{seprho}
   \rho = \sum_i p_i \rho_i \otimes \sigma_i  \,,
\ee
where $\rho_i,\sigma_i$ are states in the subsystems $A,B$ respectively and $\sum_i p_i =1$ with $p_i\geq 0$, can be prepared from non-entangled pure states $|\psi \rangle_A\otimes |\phi\rangle_B$ by separable operations (\ref{separable}), which can easily be seen using
an alternative representation of (\ref{seprho}) as an ensemble of factorized pure states\footnote{with $|\chi_j\rangle_A = A_j\ket{\psi}_A,\ |\eta_j\rangle_B=B_j \ket{\phi}_B$.}
\be\label{seprhopure}
\rho = \sum_j \tilde{p}_j \,|\chi_j \rangle_A \mbox{}_A\langle \chi_j |\otimes |\eta_j \rangle_B \mbox{}_B\langle \eta_j |\,,
\ee
which follows from ensemble decompositions of $\rho_i,\sigma_j$ and index relabeling.  This leads to the alternative definition of $\rho$ being entangled if and only if it is not separable. Entangled states then act as a resource for LOCC processes. For the simplest case, attempting to convert a pure state $\ket{\Psi}\in \mathcal{H}_{AB}$ to  
another state $\ket{\Phi}\in \mathcal{H}_{AB}$, Nielsen's theorem \cite{nielsen} provides a necessary and sufficient criterion for the possibility of state transition. 
It states that it is possible to convert $\ket{\Psi}$ to $\ket{\Phi}$ by LOCC with reference to a bipartition $AB$ , namely
%\be\label{LOCC}
$\ket{\Psi} \underset{\textrm{\tiny LOCC}}{\longrightarrow} \ket{\Phi}$, 
%\ee
if and only if the majorization condition 
%\be
$\ket{\Phi} \succ \ket{\Psi}$ is fulfilled, 
%\ee
or equivalently for the corresponding reduced states, $  \rho_A \succ \sigma_A$. 
%\ee
%%%%%%%%%%%%%% some edits by esko Feb 21 here %%%%%%%%%%%%%%%%%%%%%
 Unfortunately, the majorization condition (\ref{majorisation}) is somewhat inconvenient to verify.
%: majorization means the vector majorization $\vec{\lambda} \succ \vec{\mu}$ for the corresponding
%eigenvalue vectors (entanglement spectrum vectors) of $\rho_A, \sigma_A$, which in turn means the following. We first assume that we have already arranged the two vectors into an ordered form, such that 
%$\lambda_1 \geq \lambda_2 \geq \cdots \lambda_d$ and likewise $\mu_1 \geq \mu_2 \geq \cdots \mu_d$ for the $d$ components.
 %We then must establish inequalities for the partial sums,
%\be
%  \sum^k_{i=1} \lambda_i \geq \sum^k_{i=1} \mu_i
%\ee
%for all $k=1,\ldots ,d$. 
The task of first finding all the eigenvalues and then comparing all the partial sums is in general rather laborious (if not intractable). On the other hand, it is easier
to rule out the possibility of the LOCC transition.
One can consider a Schur concave function $g$, which is non-increasing under the
transition. Thus, if we find that $\Delta g\, \equiv g(\rho_A) - g(\sigma_A) > 0$, the transition is ruled out. 

More general LOCC processes, where a pure or mixed state is converted to a mixed state $\rho  \underset{\textrm{\tiny LOCC}}{\longrightarrow} \sigma=\Lambda (\rho)$, where $\rho$ and $\sigma$ are general (pure or mixed) states in the composite system $A\cup B$, have no simple characterization by majorization. Moreover, Schur concave functions are in general not monotonic under such processes. This leads one to consider entanglement monotones. The key requirement for monotonicity is $E(\rho) \geq E(\Lambda (\rho))$. 
In more detail, an entanglement monotone $E(\rho)$ is defined as a map $\rho \mapsto E(\rho)\in \mathbb{R}$ which satisfies  \cite{Plenio:2007zz}
\begin{enumerate}
\item $E(\rho) \geq 0$
\item$ E(\rho )=0$ if $\rho$ is separable
\item $E(\rho)$ does not increase on average\footnote{This allows for the possibility that some $E(K_i\rho K^\dagger_i)$ in the sum may increase.} under LOCC, which means
\be
E(\rho) \geq \sum_i  p_i E\left( \frac{K_i \rho K^\dagger_i}{\Tr (K_i \rho K^\dagger_i) }\right)   \,,
\ee
where $K_i = A_i\otimes B_i$ are Kraus operators of a LOCC process as in (\ref{separable}), and $p_i = \Tr (K_i \rho K^\dagger_i)$.
\end{enumerate}

Central to this work is Vidal's theorem \cite{Vidal}, which provides a way to construct monotones from concave quantities of the type (\ref{def Tr funct}). 
Consider a pure state $|\Psi\rangle$ and $\rho_A=\textrm{Tr}_B |\Psi\rangle\langle\Psi|$, and
any function $g(\rho_A)$ such that
\begin{enumerate}
\item $g$ is concave
\item $g$ is invariant under unitary transformations $U_A$, namely $g\big(U_A\rho_A U_A^\dagger\big)=g(\rho_A)$
\end{enumerate} 
then Vidal's theorem establishes first that $E_{\rm pure}(\ket{\Psi}\langle \Psi |)\equiv g(\rho_A)$ is an entanglement monotone for pure states (a pure state entanglement monotone). Moreover, one can extend $E_{\rm pure}$ to a monotone $E$  for mixed states by using the convex-roof extension, which is defined as follows. Given the density matrix 
$\sigma$, one considers the minimum over all of its ensemble decompositions $\{p_j,|\Psi_j\rangle\}$ realizing $\sigma=\sum_j p_j |\Psi_j\rangle\langle\Psi_j|$ and defines
\be
\label{convexroofextension}
E(\sigma)\,\equiv \!\min_{\{p_j,|\Psi_j\rangle\}} 
\sum_j p_j \,
E_{\rm pure} \big(|\Psi_j\rangle\langle\Psi_j| \big) = \!\min_{\{p_j,|\Psi_j\rangle\}} 
\sum_j p_j \, g\big(\textrm{Tr}_B |\Psi_j\rangle\langle\Psi_j|\big) \ .
\ee 
One can then show that $E$ is an entanglement monotone\footnote{A constant term may be added to $g$ to ensure that $E(\rho)=0$ for separable states.} \cite{Vidal,Plenio:2007zz, 4xHorodecki}. A subtle feature of this construction is that while 
%$E_{\rm pure}$ and 
$g$ is concave with respect to states for the $A$-system, $E$ is convex with respect to states on the $AB$ system (as the name "convex roof extension" implies). While our construction of entanglement monotones follows the above steps, in this work we do not need to explcitly use the convex roof extension, since  we consider only LOCC processes between pure states or processes with unital channels, both implying majorization,
where Schur concavity is sufficient to give monotonicity. Since as we mentioned earlier in this section, 
concavity of $g$ also implies that it is
Schur concave, it can therefore directly be used to find necessariy criteria for the existence of either type of process.

%While our construction of entanglement monotones follows the above steps, in this work we do not need to explcitly use the convex roof extension, since  we consider only LOCC processes between pure states or processes with unital channels, both implying majorization,
%where Schur concavity is sufficient to give monotonicity.

%%%%%%%%%%% end edits by esko %%%%%%%%%%%%%%%%%%%%%%%%

% Notice that this result holds for systems with a finite $d$ dimensional Hilbert space and that $M$ only depends on the initial state, so the first inequality cannot be reversed by interchanging $\rho$ and  $\sigma$.
%If the conversion from $\rho$ to $\sigma$ can be accomplished by a unital channel, then the inequality (\ref{entropycreation}) must be satisfied. 
% If not, the majorization cannot hold, and
%the state conversion by a unital channel is ruled out. 
%In Sec.\,\ref{subsec:ExtremalMonotones} 
We now move to construct an infinite sequence of entanglement monotones generalizing $M$ in (\ref{Mdefinition}) and inequalities generalizing (\ref{Landauer-ineq}) for pairs of majorizing states.

%\textcolor{blue}{\bf In this manuscript, we never use the property of SSA, which is more important in the other draft. Maybe we can move this part with the definition in the other file.}
%\\
%Finally, consider a Hilbert space $\mathcal{H}$ that can be decomposed as 
%$\mathcal{H}=\bigotimes_{i\in \mathcal{I}}\mathcal{H}_i$ 
%where $\mathcal{I}$ is a given set of indices.
%Given a density matrix $\rho $ and two subsets of $\mathcal{I}$ called $A$ and $B$, 
%let us define $\rho _A=\textrm{Tr}_{\bigotimes_{i\notin A}\mathcal{H}_i} \rho $ 
%and $\rho _B=\textrm{Tr}_{\bigotimes_{i\notin B}\mathcal{H}_i} \rho $. 
%The function $g$ is said to satisfy strong subadditivity (SSA) if it holds that
%\be
%\label{SSA def}
%g(\rho _A)+g(\rho _B)\geqslant g(\rho _{A\cup B})+g(\rho _{A\cap B})\,.
%\ee
%The von Neumann entropy satisfies the SSA (see \cite{Lieb73} for a proof). 

\section{Resource monotones and majorizing state transitions}
\label{sec:resource_monotones}

We have already reviewed the concept of entanglement monotones. A more general concept is that of resource monotones. Quantum resource theories (see \cite{Chitambar_2019} for a review) have been developed as a general framework to sharpen the distiction between the achievable and the unachievable in various classes of quantum processes. One makes the distinction between "free states", which are generated by the class of allowed quantum operations ("free operations"), and "resource states", which cannot be generated by free operations and must therefore be prepared by an external agent. As an example, entanglement cannot be created by LOCC operations (the free operations in the resource theory of entanglement), and it thus acts as resource. Generalizing the concept of an entanglement monotone, one can introduce quantifiers to track the loss of a resource under free operations, resource monotones $R$ that have the following property \cite{Chitambar_2019}
\be\label{resourcemonotone}
R(\rho) \geq R(\Phi (\rho))\,,
\ee
under any free operation $\Phi$ of the resource theory. In this section we will first construct sequences of entanglement monotones, and then show how they can be applied to define more general resource monotones. As an application, we
will briefly consider (the resource theory of) quantum thermodynamics.

\subsection{An infinite sequence of entanglement monotones}
\label{sec:monotones}

It was established in \cite{boes2020} that $M$ defined in (\ref{Mdefinition}) is a pure state entanglement monotone
because it is Schur concave, and thus is monotonic for reduced density matrices under majorization. This was proven as a corollary of a more general theorem involving relative quantities, as we will discuss in Sec.\,\ref{subsec:rel-quantifiers}.
%\textcolor{blue}{[we have not given the definition of Schur concavity yet]}.
Here we present an alternative simple proof, showing
that $M(\rho)$ is concave. It is straightforward to rewrite (\ref{Mdefinition}) as 
\be
\label{M2 tr F2}
 M (\rho ) = \Tr [f(\rho)] \, ,
\ee 
where 
\be
f (x) = x \left[-\ln x +1\right]^2 \,,
\ee
and it is simple to see that $f(x)$ is a concave function in the unit interval, namely for $x\in [0,1]$. This implies that $M (\rho )$ is concave, {\em i.e.} it satisfies 
\be
M (p\rho_1 + (1-p)\rho_2 ) \geq p M (\rho_1) + (1-p) M (\rho_2)\,,
\ee
for any pair of density matrices $\rho_1,\rho_2$ and for all $p\in [0,1]$. Concavity in turn implies Schur concavity, the property of $M$  proven in \cite{boes2020}. Furthermore, 
by Vidal's theorem \cite{Vidal} concavity implies that $M(\rho )-1$ can be extended by the convex-roof extension to a proper entanglement monotone for all states. 
%At this point we would like to give this quantity a name. 
Since $K=-\ln \rho$ is called modular Hamiltonian \cite{Haagbook}
%\footnote{Consider a discrete classical probability distribution with probability $p_i$ for each value $x_i$ of a random variable $X$. If $p_i$ is very low (say, the probability of a royal flush in a poker hand), then $-\log p_i$ is very large: low probability of an %outcome is associated with a large modular Hamiltonian. }
and we shift it by a constant 1, we call $M$ as the 
{\em second moment of shifted modular Hamiltonian}.  

It is straightforward to find concave generalizations of $M$ involving higher cumulants. 
Up to addition of an overall constant term which has no effect
to concavity, we define
\be
  f_n (x) \,=\, 
  x\,(-\ln x  + b_n)^n 
  =\,
   (-1)^n \, x \,\big(\ln x - b_n\big)^n  \ .
\ee
Since (for $n\geq 2$) $f''(x)= x^{-1}[n(b_n-\ln x)^{n-2}(\ln x + n-1 -b_n)]\leq 0$
for
\be 
b_n \geq n-1 \ ,
\ee
this parameter range ensures that $f_n(x)$ is concave over the unit interval for $n\in \mathbb{N}_+$. We then define a concave quantity  (and by Vidal's theorem, a pure state entanglement monotone)
\be
\label{Mn Tr Fn}
 M^{(n)}(\rho; b_n) \,=\,  \Tr \big[ f_n (\rho)\big]-b_n^n \,.
\ee
We call $M^{(n)}$ as the $n^{\rm th}$ {\em moment of shifted modular Hamiltonian}. The subtraction of the constant $b_n^n$ in (\ref{Mn Tr Fn}) ensures that $M^{(n)}(\rho_{\textrm{\tiny pure}}; b_n)=0$, when $\rho_{\textrm{\tiny pure}}$ describes a pure state.
%\footnote{In 
%Sec.\,\ref{sec:SymResCapacity} we will use a normalized version 
%\be
%   M^{(n)} (\rho;b_n) = \Tr [\rho (-\ln \rho +b_n)^n] -b^n_n\,,
%\ee
%where the subtraction of  $b^n_n$ allows to find that $ M^{(n)}(\rho_p;b_n)=0$ for pure states $\rho_p$.}. 
%We continue the discussion of the properties of $M^{(n)}$ here without the additional constant term.
When $n=2$ and $b_n=1$, the expression (\ref{Mn Tr Fn}) reduces up to an additive constant to (\ref{M2 tr F2}). 
Moreover, since $M^{(n)}$ has the form (\ref{def Tr funct}) with $f_n$ concave in $[0,1]$, 
we can conclude that it is also Schur concave,
from the discussion in Sec.\,\ref{subsec:concavity}.

To rewrite (\ref{Mn Tr Fn}) as combination of cumulants, we expand it first as a linear combination of moments $\mu_k (\rho)$ of modular Hamiltonian
\be
\label{lincombmoments}
   M^{(n)}(\rho;b_n) 
   %= \sum^n_{k=0} {n \choose k} b_n^{n-k} \Tr \left[ \rho (-\ln \rho)^k \right]
   = \sum^n_{k=1} {n \choose k} b_n^{n-k} \mu_k (\rho) \, ,
   \;\;\;\qquad\;\;\;
   \mu_k=\textrm{Tr}\big[\rho \,(-\ln \rho )^k\big]
   =
   \textrm{Tr}\big( \rho \,K^k\big)\,.
\ee
Then, in turn we can use the relation between moments and cumulants (see e.g. \cite{kardar_2007}) to write
\be\label{momentstocumulants}
  \mu_k (\rho) = \sum'_{\{p_j\}} k! \prod_j \frac{1}{p_j! (j!)^{p_j}} \;C_j (\rho)^{p_j}\, ,
\ee
where the restricted sum $\sum'$ is over all partitions $\{p_j\}$ of $k= \sum_j j\, p_j$, and $C_j(\rho)$ is the $j^{\rm th}$ cumulant of modular Hamiltonian. 
%\be
%C_j (\rho) = \Tr (\rho (\ln \rho - S(\rho))^k) \ .
%\ee
A more streamlined way is to expand the moments $M^{(n)}(\rho ;b_n)$ as cumulants $\tilde{C}_j (\rho)$ of $-\ln \rho + b_n$ as follows
\be\label{momentstocumulants2}
  M^{(n)} (\rho ;b_n) = \sum'_{\{p_j\}} n! \prod_j \frac{1}{p_j! (n!)^{p_j}} \;\tilde{C}_j (\rho)^{p_j}-b_n^n\, ,
\ee
where $n=\sum_j j\, p_j$.  
In the sum over partitions, terms involving cumulants $\tilde{C}_j$ of $-\ln\rho + b_n$
 of order $j\geq 2$ reduce to $C_j (\rho)$ due to translation invariance. Terms involving $j=1$ give $\tilde{C}^{p_1}_1 (\rho)=[S(\rho)+b_n]^{p_1}$. In this way, setting for simplicity $b_n$ to be the smallest possible value $b_n=n-1$ for concavity, we obtain the sequence 
\begin{eqnarray}\label{extrM_variousn}
M^{(1)}(\rho;0) &=& S(\rho)\,,  \\
M^{(2)}(\rho ;1) &=& [S(\rho)+1]^2 + C(\rho )-1\,,  \nonumber \\
M^{(3)}(\rho ;2) &=& [S(\rho)+2]^3  + 3 C (\rho) [S(\rho) +2] + C_3 (\rho)-2^3\,, \nonumber \\
M^{(4)}(\rho ;3) &=& [S(\rho) + 3]^4 + 6 C(\rho )[S(\rho) +3]^2 + 3 C(\rho)^2 + 4 C_3 (\rho)[S(\rho) + 3] + C_4 (\rho) -3^4\,,\nonumber \\
 &\cdots& \,. \nonumber
\end{eqnarray}
The cumulants of modular Hamiltonian can be derived from a generating function as \cite{deBoer:2018mzv, boes2020}
\be
\label{CumulantGenerating}
C_n (\rho) =  (-1)^n \frac{d^n}{d\alpha^n} \big[(1-\alpha )\, S^{(\alpha )} (\rho)\big]\Big|_{\alpha=1} \,,
\ee
where $S^{(\alpha )} (\rho)$  are the R\'enyi entropies (\ref{defRenyi}). 
The R\'enyi entropies in turn are determined by the entanglement
spectrum. Note that the R\'enyi entropies themselves are not concave when the  R\'enyi index $\alpha>1$. On the other hand, the full information about the bipartite entanglement is encapsulated by the entanglement spectrum, 
which provides the R\'enyi entropies. 
They in turn can be converted to the cumulants, which can be converted to the 
entanglement monotones $M^{(n)}$. To summarize, the above sequence provides a way to convert the full entanglement spectrum into an infinite sequence of entanglement monotones.
It is helpful to note that the cumulants of modular Hamiltonian are additive, namely
\be
\label{additivity_Cn}
C_n (\rho_1 \otimes \rho_2 ) \,=\,  C_n (\rho_1 ) + C_n (\rho_2 )\, .
\ee
This can be verified in different ways, most simply it follows from the additivity of the R\'enyi entropies and the generating function formula (\ref{CumulantGenerating}) for cumulants. 
Also, for the maximally mixed state $\rho = \mathbb{1}/d$ of a system with a $d$-dimensional Hilbert space, 
where $\mathbb{1}$ is the $d\times d$ identity matrix,
the same formula gives a simple proof of 
\be
 C_n (\mathbb{1}/d) = \Bigg\{ 
 \begin{array}{ll} 
 \; \ln d \hspace{.6cm}& {\rm for}\ n= 1\,, 
 \\ 
 \rule{0pt}{.5cm}
\; 0 & {\rm for}\  n\geq 2  \,.
 \end{array} 
\ee
Notice that $C_n=\mu_n=0$ for pure states. Also, for maximally mixed state, we have 
\be
M^{(n)}(\mathbb{1}/d; b_n) = (\ln d + b_n)^n-b^n_n \ ,
\ee
so that one would need to rescale by an overall normalization constant if one wishes to follow the convention \cite{Plenio:2007zz} that an entanglement monotone is normalized to $\ln d$ for the maximally mixed state ({\em i.e.} $M^{(n)}(\mathbb{1}/d;b_n)= \ln d$). 

In fact, the sequence $M^{(n)}$ can be derived more straightforwardly from the R\'enyi entropies, converting the latter to a generating function as follows. 
Since
\be
\Tr ( \rho^\alpha ) = \exp [(1-\alpha )S^{(\alpha )}(\rho )]\, ,
\ee
in terms of the R\'enyi entropies $S^{(\alpha )}(\rho)$, by defining
\be
k_{\alpha}(\rho; b) 
\,\equiv\, 
e^{-\alpha b} \,\Tr (\rho^\alpha) 
\,=\, 
\Tr \big[e^{\alpha (\ln \rho - b)} \big]
= \exp \big[\! -\alpha b + (1-\alpha)S^{(\alpha )} (\rho)\big] \,,
\ee
we have that
\be
\label{Mn from Sn}
 M^{(n)}(\rho ;b_n ) = \left[e^{b}(-1)^n  \frac{d^n}{d\alpha^n} \;k_\alpha (\rho ;b)\right]\bigg|_{\alpha =1, b=b_n}-\,b_n^n \, ,
\ee
which provides a prescription for converting the R\'enyi entropies to the  infinite sequence of entanglement monotones.

To summarize: the complete information of the bipartite entanglement is encapsulated by the entanglement spectrum and this information can be repackaged first to the set of R\'enyi entropies, 
which in turn can be converted to an infinite tower of entanglement monotones $ M^{(n)}(\rho;b_n)$. 

Note that concavity only gives a lower bound $b_n\geq n-1$. Other conditions
%, such as the application of $ M^{(2)}$ in the Landauer inequality, 
may lead to a particular choice for the value of the constant $b_n$.%, {\em e.g.} $b_2=1/\ln 2$ in the Landauer inequality in \cite{boes2020}.  

In Appendix \ref{subsec:highercumulants} we report additional comments on how to construct entanglement monotones exploiting the cumulants of the modular Hamiltonian.

\subsection{Extremal polynomial monotones and inequalities for state transitions}
\label{subsec:ExtremalMonotones}

In this section we perform a more general analysis of infinite sequence of monotones constructed from convex polynomials of the moments of $\ln \rho$. With some abuse of terminology, we will refer to them as "polynomial entanglement monotones" for brevity. 
We find that such monotones form an infinite dimensional cone, which
is determined by extremal rays, defining what we will call for brevity {\em extremal polynomial monotones}. The extremal polynomial monotones give rise to an infinite sequence of inequalities that
must be satisfied in majorizing state transformations. 

Our starting point is the general functional
\be
\label{general polynomial monotone}
P(\rho) = {\rm Tr}\big[\rho\, F(\ln\rho)\big] \,.
\ee
The corresponding scalar function $f(x)=x\,F(\ln x)$ is convex if $f''\geq 0$ for $x\in [0,1]$. 
For convenience, we focus on convex measures, which can be converted to concave measures by a minus sign.  
This translates into the condition
\be\label{possemidef}
F'(y) + F''(y)\geq 0\,, \qquad y\leq 0\,.
\ee
For example, $f(y)=y$ clearly meets this criterion, with $P$ in (\ref{general polynomial monotone}) being minus the von Neumann entropy. Convex functions yield convex measures
which are monotonic under majorization (Schur convex)
\be
\rho \succ \sigma 
\qquad
\Longrightarrow
\qquad
P(\rho) \geq P(\sigma) \ .
\ee

%\section{A more general analysis}

In the previous subsections, we have seen that for suitable functions $F$ we get monotones which we use to test for majorization. 
We now want to be more systematic and classify all $F$ with the property that $F'+F''\geq 0$ for $y\leq 0$. We restrict our survey here by focusing on measures
where $F$ is polynomial in $y=\ln \rho$. Let us introduce 
\be
G(y)\equiv F'(y)+F''(y)\,.
\ee
From (\ref{possemidef}), consider all polynomials $G(y)$ such that $G(y)\geq 0$ for $y\leq 0$. For each 
such polynomial $G$, there is a unique polynomial $F$ such that $F''+F'=G$ (up to vanishing constant terms, which can be added at will {\em e.g.} changing the value of $F$ and $P$ for pure states). 

The space of positive semidefinite polynomials $G(y)$ on negative real axis 
%with the property that $G \geq 0$ for $y\leq 0$
is a convex cone  in the sense that, given a
set of functions $G_i$ with this property, a linear combination $\sum_i \alpha_i G_i$ with non-negative coefficients will also
have this property\footnote{We remark that there exists a different line of investigation that uses convex cones, to classify and constrain entropy inequalities in holographic
gauge-gravity duality, initiated in \cite{Bao:2015bfa}.}. Cones are completely determined by specifying all "extremal" rays, which are functions $G$ which
do not admit a non-trivial decomposition of the type $G=\sum_i \alpha_i G_i$. The most general $G$ will then be a linear
combination of extremal functions $G$ with non-negative coefficients. In general, there can be finitely or infinitely 
many extremal $G$. 

From the perspective of monotones, the extremal $G_j$ will provide a complete list of non-trivial "extremal polynomial monotones", with all
other polynomial monotones being linear combinations of extremal polynomial monotones with non-negative coefficients. It is therefore 
interesting to classify all such extremal monotones. For this we need to classify all extremal $G_j$. We can use known results from the 
theory of positive semidefinite polynomials, which can be summarized as the following theorem.

\begin{thm} All positive semidefinite polynomials $G(y)$ on the negative half-line $y\in (-\infty ,0]$ have the following form. 
For polynomials $G(y)$ of degree $2d$ (with $d\geq 1$), they are linear combinations with non-negative coefficients of polynomials of the form 
$G_{\vec{a}}(y)= \prod^d_{i=1} (y+a_i)^2$ with all $a_i\geq 0$. For polynomials of degree $2d+1$ they are linear combinations with non-negative coefficients of  polynomials of the form
$G_{\vec{a}}(y)=-\,y\,\prod^d_{i=1} (y+a_i)^2$ with again all $a_i\geq 0$.
\label{thm_Pol}
\end{thm}

%\begin{proof}
We defer the detailed proof to Appendix \ref{app:Proof}. 
The result is essentially known in mathematics (see \cite{PowRez} for a review of non-negative polynomials).
%\end{proof}

We emphasize that the higher moments of modular Hamiltonian $M^{(n)}$, introduced in Sec.\,\ref{sec:monotones}, are in general not extremal monotones. Consider $F(x)=(x-b_k)^k$ so that $G(x)=F'+F''= k(x-b_k)^{k-2}(x-b_k+k-1)$. This has an isolated zero at $x=b_k-k+1$. Since $b_k \geq k-1$, this isolated zero cannot be on the negative real axis. In the limiting case $b_k=k-1$ one obtains $G=kx(x-k+1)^{k-2}$. This has many zeroes on the positive real line rather than
even degeneracy zeroes on the negative real line, hence as such it is not extremal. Thus the moments $M^{(n)}(\rho ;b_n)$ in general are linear combinations of the extremal monotones $P^{(k)}_E (\rho)$. 

Let us now study the lowest degree examples in detail. For $F$ of degree 1, $G$ has degree zero and must be a non-negative constant, which we can take to be $1$. 
Then $F(y)=y$ and the resulting extremal monotone $P^{(1)}_E$ is minus the entropy, namely $P^{(1)}_E=-S(\rho)=-M^{(1)}$. 

For $F$ of degree 2, $G$ is of degree one. According to the Theorem \ref{thm_Pol}, there is a unique extremal $G$ which is $-y$. Solving $F''+F'=-y$, we have $F=y-y^2/2$. We might as well take twice this as extremal functions are defined up to overall
normalization only. In that case $F=2y-y^2$. Thus, we have proven that 
\be
P^{(2)}_E(\rho) ={\rm Tr} \big[\rho (2\ln \rho -\ln^2 \rho)\big]
=-C(\rho)-S(\rho)^2 - 2S(\rho) = -M^{(2)}(\rho, b_2=1) 
\ee
is an extremal monotone. Note that, up to second order, the two classes of monotones are related by $M^{(n)}(\rho;n-1) = -P_E^{(n)}(\rho)$, while this is no longer
true for $n\geq 3$.
% If we only consider functions involving one or two powers of $\ln\rho$, so that they can
%be expressed in terms of variance and the entropy only, the only non-trivial monotones are $P^{(1)}_E=-S$ and $P^{(2)}_E$, with  other ones being their linear combinations
%with positive coefficients (and possibly shifted by a constant). 

Instead of the entropy production inequality with the finite correction (\ref{Landauer-ineq}), the extremal monotone appears to give a slightly sharper inequality. 
The statement 
\be
\label{j1}
\rho \succ \sigma 
\quad
\Longrightarrow
\quad
P^{(2)}_E(\rho) = -C(\rho)-S(\rho)^2 - 2S(\rho) \geq -C(\sigma)-S(\sigma)^2 - 2S(\sigma) = P^{(2)}_E (\sigma ) 
\ee
can be rewritten as the inequality
\be \label{j2}
S(\sigma) - S(\rho) \geq \frac{ C(\rho) - C(\sigma) }{ S(\rho) + S(\sigma) + 2 }\,,
\ee
which appears to be slightly sharper than (\ref{Landauer-ineq}) involving
$M$, $S$ and $C$. 

Until now, the inequalities have been the same as the ones coming from 
\be\label{Ineqsimple}
M^{(n)}(\sigma ;n-1)\geq M^{(n)}(\rho ;n-1) \,,
\ee
when $\rho\succ\sigma $. 
In terms of the cumulants, using (\ref{momentstocumulants2}), this sequence has the explicit form
\begin{eqnarray}\label{simplesequence}
 && S(\sigma)\geq S(\rho)\,, \\
 && [S(\sigma )+1]^2 + C(\sigma ) \geq [S(\rho)+1]^2 + C(\rho )\,, \nonumber  \\
 &&[S(\sigma )+2]^3  + 3 C (\sigma ) [S(\sigma ) +2] + C_3 (\sigma ) \geq [S(\rho)+2]^3  + 3 C (\rho) [S(\rho) +2] + C_3 (\rho)\,, \nonumber \\
 % (S(\sigma ) + 3)^4 + 6 C(\sigma  )(S(\sigma ) +3)^2 + 3 C(\sigma )^2 + 4 C_3 (\sigma )(S(\sigma ) + 3) + C_4 (\sigma )  &\geq& (S(\rho) + 3)^4 + 6 C(\rho )(S(\rho) +3)^2 + 3 C(\rho)^2 + 4 C_3 (\rho)(S(\rho) + 3) + C_4 (\rho) \\
 &&\cdots   \ . \nonumber
\end{eqnarray}
Notice that the "second law" of entropy (claiming that the entropy is non-decreasing in transitions $\rho \mapsto \sigma$ with $\rho \succ \sigma$) becomes refined into an infinite sequence of
inequalities that must likewise be satisfied.
However, at orders $n\geq 3$ the extremal polynomial monotones may give tighter inequalities. %Here we only demonstrate this for the case $n=3$. 

For $F$ of degree 3, $G$ is of degree two and must be of the form $(y+a)^2$ with $a>0$. So in this case we get a one-parameter family of
extremal monotones. Solving $F'+F''=G$ yields 
\be
F= \frac{1}{3} y^3 + (a-1) y^2 +(a^2-2a+2) y\,,
\ee
and this gives rise to a one-parameter family of extremal monotones for $a\geq 0$
\be
P^{(3)}_E(\rho) \,=\, 
{\rm Tr}\!\left[\rho \left( \frac{1}{3} \ln^3\rho + (a-1) \ln^2 \rho + (a^2 - 2 a + 2) \ln\rho\right)\right]\,,
\ee
and correspondingly to an infinite number of inequivalent inequalities. It is useful to express the coefficients of $a^k$ in terms of the monotones $M_n\equiv M^{(n)}(\rho; n-1)$. Let us first denote $r^{(k)}\equiv{\rm Tr}[\rho \left(-\ln \rho\right)]$ and $r\equiv r^{(1)}$. 
We have (including 4th order for future reference)
\bea
&& M_1 = \Tr [\rho (-\ln \rho)] = r\,, \\
&& M_2 = \Tr \left[\rho (-\ln \rho +1)^2\right] -1= r^{(2)} +2\,r\,,  \\
&& M_3 = \Tr \left[\rho (-\ln \rho +2)^3\right]-2^3 = r^{(3)} + 6\,r^{(2)} +12 \,r\,, \\
\label{rtom4}
&& M_4 = \Tr \left[\rho (-\ln \rho +3)^4\right]-3^4 = r^{(4)} + 12\, r^{(3)} + 54\,r^{(2)} + 108 \,r \, .
\eea
We can then express
\bea
P^{(3)}_E(\rho) &=& -a^2 r + a (r^{(2)} +2r) - \frac{1}{3}r^{(3)} -r^{(2)} -2r \nonumber \\
    \mbox{}&=& -a^2 M_1 +a M_2 -\frac{1}{3} M_3 +M_2  \,.
\eea
We could thus calculate $P^{(3)}_E$ from the R\'enyi entropies by using first the generating function formula (\ref{Mn from Sn}) for $M^{(n)}$.

We explore the cubic case a bit more, to find a tight inequality. Assuming $\rho \succ \sigma$, we obtain 
\be\label{p3inequality}
P^{(3)}_E(\rho)\geq P^{(3)}(\sigma) \ ,
\ee
which at this stage is an infinite family of inequalities due to the free parameter $a$.
The inequality (\ref{p3inequality}) can be written more explicitly as
\be \label{j11}
 w_2 a^2 + w_1 a + w_0 \equiv a^2\Delta M_1  + -a\Delta M_2+\frac{1}{3}\Delta M_3-\Delta M_2 \geq 0\,,
%\frac{1}{3} r^{(3)} + (1-a) r^{(2)} + (a^2-2 a +2) r \leq \frac{1}{3} s^{(3)} + (1-a) s^{(2)} + (a^2-2 a +2) s\,, \\
%\label{j12}% && \Leftrightarrow  \geq 0\,,
\ee
where
\be
\label{DeltaM n definition}
\Delta M_n = M^{(n)}(\sigma;n-1)-M^{(n)}(\rho; n-1)\geq 0 \ .
\ee
Notice that $w_2\geq 0$ while $w_1\leq 0$. 
The quadratic function of $a$ will therefore have a minimum at 
\be\label{a0}
a_0=-\frac{w_1}{2w_2}\geq 0 \ . 
\ee
In order to check whether
the inequalities (\ref{j11}) are satisfied for all $a\geq 0$, we only need to verify it for  $a_0$ for which the
quadratic polynomial in (\ref{j11}) takes its mimimum value. This leads to 
\be\label{w-inequality1} 
w_0 - \frac{w_1^2}{4 w_2} \geq 0\ .
\ee
%or equivalently (since $w_2>0$)
%\be
%4 w_0 w_2 - w_1^2 \geq 0\,.
%\ee
Finally, substituting the $w_n$ from (\ref{j11}), the inequality $(\ref{w-inequality1})$ %(or (\ref{M3inequality})) 
takes the compact form
\be\label{Ineq3compact}
\Delta M_3 \geq 3 \Delta M_2 + \frac{3}{4} \frac{(\Delta M_2)^2}{\Delta M_1} \ .
\ee
Now we can explicitly see the advantage of (\ref{Ineq3compact}) over the simple inequality $\Delta M_3\geq 0$ from (\ref{Ineqsimple}). We could then
substitute the explicit forms of $\Delta M_n$ from (\ref{DeltaM n definition}) to explore how entropy production and changes in other cumulants up to third order are bounded by each other. 
Alternatively, we can study differences in moments $r^{(n)}$ and $s^{(n)}\equiv \Tr [\sigma (-\ln \sigma )^n]$. For example, the difference $r^{(3)}-s^{(3)}$ has a lower bound in terms of $r,r^{(2)},s$ and $s^{(2)}$
\be\label{M3inequality}
r^{(3)}-s^{(3)} \leq -3(r-s) +\frac{3}{4}\frac{\left(r^{(2)}-s^{(2)}\right)^2}{r-s}\, .
\ee
%\be\label{M3inequality}
%r_3-s_3 \geq 3(r_2-s_2) - 6(r_1-s_1) + \frac{3 ( (r_2-s_2)- 2 (r_1-s_1) )^2 }{4 (r_1-s_1) } .
%\ee
 Note that (\ref{Ineq3compact}) has
an interesting hierarchy where the third order $\Delta M_3$ is bounded by  are bounded by combinations $\Delta M_k$ with $k=1,2$, or the next order inequality is bounded by
the previous order inequalities. 
This suggests that higher order inequalities could have some interesting recursive structure. We explore this some more by working out the
fourth order monotone $P^{(4)}_E$ and the resulting inequalities.

To find $P^{(4)}(\rho)$, we first need to solve $F''(y)+F'(y)=G(y)$ with
\be
  G(y) =-y(y+a)^2 \,,
\ee
where $y=\ln \rho$ and $a>0$. The polynomial solution is
\be
F(y)=a^2 \left(y-\frac{y^2}{2}\right)+a\left(-4y+2y^2 -\frac{2y^3}{3}\right)+6y-3y^2+y^3-\frac{y^4}{4} \ ,
\ee
collecting coefficients of $a^k$. Then, $P^{(4)}_E (\rho)=\Tr [\rho F(\rho)]$, whose explicit form reads
\be
P^{(4)}_E=a^2 \left(-r-\frac{r^{(2)}}{2}\right)+a\left(4r+2r^{(2)} +\frac{2r^{(3)}}{3}\right)-6r-3r^{(2)}-r^{(3)}-\frac{r^{(4)}}{4} \ .
\ee
The coefficients of $a^k$ can be expressed in terms of the monotones $M_n\equiv M^{(n)}(\rho; n-1)$  and using (\ref{rtom4}). 
%With the above notation for the expectation values, we have
%\bea
%&& M_1 = \Tr [\rho (-\ln \rho)] = r\,, \\
%&& M_2 = \Tr \left[\rho (-\ln \rho +1)^2\right] -1= r^{(2)} +2r\,,  \\
%&& M_3 = \Tr \left[\rho (-\ln \rho +2)^3\right]-2^3 = r^{(3)} + 6r^{(2)} +12 r\,, \\
%&& M_4 = \Tr \left[\rho (-\ln \rho +3)^4\right]-3^4 = r^{(4)} + 12 r^{(3)} + 6 \cdot 9r^{(2)} + 4\cdot 3^3 r \, .
%\eea
With some calculation, we find
\bea
&&-\,r-\frac{r^2}{2} = -\half M_2\,, \\
&& 4r+2r^{(2)} +\frac{2r^{(3)}}{3} = \frac{2}{3}\left(M_3-3M_2\right)\,, \\
&&-\,6r-3r^{(2)}-r^{(3)}-\frac{r^{(4)}}{4} = -\frac{1}{4}(M_4-8M_3 +6M_2) \, ,
\eea
so that 
\be
P^{(4)}_E (\rho) = \left(-\half M_2\right)a^2 + \left[\frac{2}{3}(M_3-3M_2)\right]a  -\frac{1}{4}(M_4-8M_3 +6M_2) \ .
\ee
Now consider a pair of majorizing states. As before, we have
\be\label{ineq}
\rho \succ \sigma\qquad \Rightarrow \qquad P^{(4)}(\rho)\geq P^{(4)}(\sigma) \ . 
\ee
Let us denote $\Mtilde_n \equiv M^{(n)} (\sigma ;n-1)$ and $\Delta M_n = M^{(n)}(\sigma;n-1)-M^{(n)}(\rho; n-1)\geq 0$, as in (\ref{DeltaM n definition}). As before, we rewrite the inequality (\ref{ineq}) in the form
\be\label{aineq}
w_2 a^2 + w_1 a + w_0 \geq 0  \,,
\ee
where now
\bea\label{ws}
&&w_2 \equiv\half \left(\Mtilde_2-M_2\right) \equiv \half \Delta M_2 \geq 0 \,, \\
&& w_1 \equiv -\frac{2}{3}\left(\Delta M_3-3\Delta M_2\right) \leq 0\,, \\
&& w_0 \equiv \frac{1}{4} \left(\Delta M_4 -8\Delta M_3 + 6\Delta M_2\right) \, .
\eea
Note that in the above $w_1 \leq 0$ since $\Delta M_3 \geq 3\Delta M_2$ by the inequality (\ref{Ineq3compact}). We then find the minimum of (\ref{aineq}) at $a_0$ as in (\ref{a0})
and obtain the same $w$-inequality (\ref{w-inequality1}) as before. Substituting the $w_k$ from (\ref{ws}) by $\Delta M_n$
%\be
%a_0 = -\frac{w_1}{2w_2} \geq 0 \,,
%\ee
%and substituting back we get the same inequality as before,
%\be
%w_0 \geq \frac{w^2_1}{4w_2} \ .
%\ee
%Writing this with $\Delta M_n$ 
gives the final form of the $P^{(4)}_E$ inequality:
\be\label{P4ineq}
\Delta M_4 \geq 8 \Delta M_3 -6\Delta M_2 + \frac{8}{9}\frac{(\Delta M_3-3\Delta M_2)^2}{\Delta M_2} \ .
\ee
The right hand side is positive since
\be
8 \Delta M_3 - 6\Delta M_2 = 6 \Delta M_3 + 2 (\Delta M_3 -3\Delta M_2) \geq 0\,,
\ee
from (\ref{Ineq3compact}). We could use (\ref{Ineq3compact}) twice in the right hand side of (\ref{P4ineq}) to relax the lower bound a bit to the inequality  
\be
\Delta M_4 \geq 6\Delta M_3 + \frac{3}{2}\frac{(\Delta M_2)^2}{\Delta M_1} + \half \frac{(\Delta M_2)^3}{\Delta M_1} \ .
\ee
The upshot is that the 4th order inequality is manifestly tighter than just $\Delta M_4\geq 0$. There is still some interesting recursive structure, albeit the terms appearing in the right hand side of (\ref{P4ineq}) do not appear in the same combinations as in the previous order inequalities:
both $\Delta M_3$ and $\Delta M_3 - 3\Delta M_2$ appear instead of only $\Delta M_3 - 3\Delta M_2 -\frac{3}{4}\frac{(\Delta M_2)^2}{\Delta M_1}$.

 The virtue of $M^{(n)}(\rho ;b_n)$ is that
they are simple to compute (for example when the R\'enyi entropies are known) and provide inequalities involving cumulants up to order $n$. The inequalities $\Delta M^{(n)}\geq 0$ are expected to be weaker than those derived from the extremal polynomial monotones
$P^{(n)}_E(\rho)$. A trade-off is that the latter inequalities are less straightforward to derive, due to the increasingly many free parameters contained in $P^{(n)}_E$, which need to be optimized to make the inequalities as tight as possible.

\subsection{Other resource monotones from pure state entanglement monotones}
\label{subsec:rel-quantifiers}

In Ref.\,\cite{boes2020} a more general version of majorization has also been considered, which has applications to other resource theories than that of entanglement, {\em e.g.} to quantum thermodynamics. Let ${\cal D}$ denote the set of quantum states of 
a given system. Consider two pairs of states $\rho ,\sigma \in {\cal D}$ and $\rho' ,\sigma' \in {\cal D}$. If there exists a quantum channel $\Ecal$ in $\cal{D}$ such that $\Ecal (\rho)=\rho' $ and $\Ecal (\sigma)=\sigma'$, we denote $(\rho ,\sigma) \succ (\rho',\sigma')$, defining a partial order between pairs of states. A special case is the fixed point $\sigma = \sigma'\equiv \sigma_*$ of the free operations, we then write $\rho \succ_{\sigma_*} \rho'$ instead of $(\rho ,\sigma_*)\succ (\rho' ,\sigma_*)$ and say
that $\rho$ $\sigma_*${\em -majorizes} $\rho'$. An important example is the {\em thermomajorization}, where the fixed point is the Gaussian thermal state, $\sigma_* = e^{-\beta H}/Z$ (or a generalized
Gaussian state). This leads to a partial order in quantum thermodynamics.  For the generalized majorization $(\rho ,\sigma) \succ (\rho',\sigma')$, the relative entropy $S(\rho ||\sigma)$ is monotonic with 
\be\label{Smonotonic}
S(\rho || \sigma) \geq S(\rho'||\sigma')\,,
\ee
 and is called a resource monotone in the above context. The monotonicity (\ref{Smonotonic}) is a special case of the more general contractivity property $S(\rho ||\sigma)\geq S({\cal N}(\rho )||{\cal N}(\sigma ))$ for any quantum channel ${\cal N}$.

In Ref.\,\cite{boes2020} a new relative quantifier was introduced which takes the form
\be\label{relM2}
M_x(\rho ||\sigma ) = C (\rho || \sigma) + \big( 1 -\ln (x) - S(\rho || \sigma )\big)^2,
\ee
%\textcolor{red}{Recheck this equation, in particular, should $\log_2 (x)$ get a factor of $\ln2$? }
where $\sigma $ is a full rank state, $C(\rho || \sigma)$ is the variance of  the relative modular Hamiltonian\footnote{Note that \cite{boes2020} defines $\tilde M_x(\rho ||\sigma ) = \tilde C (\rho || \sigma) + \left( \frac{1}{\ln 2} -\log_2 (x) -\tilde  S(\rho || \sigma )\right)^2$, where
we use the tilde notation to emphasize that the quantities are based on binary logarithms, they involve rescalings by $\ln 2$, so that $M_x(\rho || \sigma)=(\ln 2)^2 \tilde M_x(\rho||\sigma)$, to compare with (\ref{relM2}).},
\be
C(\rho ||\sigma) = \Tr [\rho (\ln\rho -\ln \sigma )^2] - S(\rho || \sigma)^2\,,
\ee
where $S(\rho || \sigma)\equiv \mathrm{Tr} [\rho (\ln \rho -\ln \sigma )]$ is the relative entropy, which is the expectation value of the relative modular Hamiltonian. They considered a setting of "pairs of classical states", that is assuming that $[\rho ,\sigma]=[\rho',\sigma']=0$, further assuming that $\sigma ,\sigma'$ are both full rank. In this setting, they proved some interesting properties for the quantifier (\ref{relM2}), and proved a lower bound  for the production of relative entropy, involving variations of the relative variance.  
In this section we generalize this construction of \cite{boes2020} to an infinite class of relative quantifiers, which are also limited to pairs of classical states, i.e. pairs of mutually commuting states. A more general analysis in the non-commuting "quantum" case remains an important open problem.

Consider first quantifiers of the form (\ref{general polynomial monotone}) %(we are now reverting back to the convention of using $\ln$ instead of $\log_2$),
\be\label{Erho}
E(\rho) = \Tr [\rho F(\ln \rho)]\,,
\ee
where we change the notation with respect to (\ref{general polynomial monotone}), to emphasize that now we allow $F(x)$ to be any smooth function (not just a polynomial), such that $xF(x)$ is concave in the unit interval $x\in [0,1]$. 
By Vidal's theorem, (\ref{Erho})
defines a pure state entanglement monotone in a bipartite system $A \cup B$, when $\rho$ is the reduced state of a pure state $|\psi\rangle_{A\cup B}$. Then, for a pair of commuting full rank density matrices $\rho ,\sigma$, we define a relative quantifier
\be\label{Exrel}
E_x(\rho || \sigma) =  E(x\rho\sigma^{-1} ) = \Tr [\rho F(\ln (x\rho\sigma^{-1}))] = \Tr [\rho F(\ln \rho -\ln \sigma +\ln (x))]\,,
\ee  
where $x$ is a real number. We then prove the following  theorem, which generalizes Theorem 12 of \cite{boes2020}:
\begin{thm}
Let $(\rho ,\sigma)$ and $(\rho',\sigma')$ be two pairs of commuting states, that is $[\rho,\sigma]=[\rho',\sigma']=0$, and $\sigma ,\sigma'$  both full rank. If $(\rho ,\sigma)\succ (\rho' ,\sigma')$, namely $\rho \succ \rho'$ and $\sigma\succ \sigma'$, then 
\be\label{Eineq}
 E_{s_{\rm min}}(\rho'||\sigma') \geq E_{s_{\rm min}} (\rho || \sigma )\,,
\ee
 where $s_{\rm min}$ denotes the smallest eigenvalue of $\sigma$.
\label{relativethm}
\end{thm}
\begin{proof}
The proof is a simple modification of the proof reported in Appendix G of \cite{boes2020}, so we will only present the essential steps here, referring to \cite{boes2020} for details.

 First, since $[\rho , \sigma]=0$ we can diagonalize both states in the same eigenbasis and write $\rho = \sum_i r_i |i\rangle\langle i|$ and $\sigma = \sum_i s_i |i\rangle\langle i|$.  Then we write
\bea
 E_{s_{\rm min}} (\rho || \sigma ) &=& \sum_i r_i F\left(\ln \left(s_{\rm min}\frac{r_i}{s_i}\right)\right) =  \frac{1}{s_{\rm min}} \sum_i s_i \left(  s_{\rm min}\frac{r_i}{s_i } \right)  F\left(\ln\left(s_{\rm min} \frac{r_i}{s_i}\right)\right) \nonumber \\
          &\equiv & \sum_i s_i g_{s_{\rm min}}\left( \frac{r_i}{s_i}\right) \ ,
\eea
where
\be
   g_{s_{\rm min}}\left( \frac{r_i}{s_i}\right) \equiv \frac{1}{s_{\rm min}}\left(  s_{\rm min}\frac{r_i}{s_i } \right)  F\left(\ln\left(s_{\rm min} \frac{r_i}{s_i}\right)\right) 
\ee
is a concave function in the interval $[\min_i \frac{r_i}{s_i}, \max_i \frac{r_i}{s_i}]$ since $xF(x)$ is concave in the unit interval, $s_{\rm min}>0$ since $\sigma$ is full rank and $s_{\rm min} \frac{r_i}{s_i} \in [0,1]$. On the other hand,
since $(\rho ,\sigma)\succ (\rho' ,\sigma')$, there exists a quantum channel mapping $(\rho,\sigma)$ to $(\rho',\sigma')$, thus by Lemma 20 of \cite{boes2020} there exists a right stochastic matrix $T$ that maps the eigenvalue vectors (denoted with bold symbols) as $\boldsymbol{r}\,T=\boldsymbol{r}'$,
$\boldsymbol{s}\,T=\boldsymbol{s}'$.
%, where $\boldsymbol{r},\boldsymbol{r}',\boldsymbol{s}$ and $\boldsymbol{s}'$.  
Then, by Lemma 34 of \cite{boes2020}, it follows that the inequality (\ref{Eineq}) holds. Note also that from the
majorization $\sigma \succ \sigma'$, it follows that the respective smallest eigenvalues satisfy $s_{\rm min} \leq s'_{\rm min}$.
\end{proof}

For an operation $\Phi$ with a fixed state $\sigma_*$ and $\rho \succ_{\sigma_*} \Phi(\rho)$, the inequality (\ref{Eineq}) implies that
\be
R_*(\rho) \equiv -E(\rho ||\sigma_*) \geq -E(\Phi(\rho)||\sigma_*) \equiv R_*(\Phi (\rho))\ , 
\ee
thus we make contact with
the definition of a resource monotone (\ref{resourcemonotone}).  
As an application of the construction (\ref{Exrel}), we can use our monotones $M^{(n)}$ and $P_E^{(n)}$ to define relative quantifiers as follows
\bea
\label{RelativeM_first def}
 M^{(n)}_{b_n,x} (\rho || \sigma ) &=& (-1)^n\Tr [\rho (\ln (x\rho \sigma^{-1})-b_n )^n] = M^{(n)}(x\rho\sigma^{-1};b_n)+b^n_n 
 %- (b_n-\ln(x))^n 
 \nonumber \\
&=& (-1)^n \Tr [\rho (\ln \rho -\ln \sigma + \ln (x)-b_n)^n] \,.
%-(b_n-\ln(x))^n
\eea
%where the last constant term has been adjusted to ensure that $ M^{(n)}_{b_n,x} (\sigma || \sigma) =0$.
The  first two quantifiers of the sequence reduce to (minus) the relative entropy and $M_x (\rho || \sigma)$ defined in (\ref{relM2}) by
\bea
% M^{(1)}_{0,1}(\rho || \sigma) &=& -S(\rho || \sigma) \\
M^{(2)}_{1, x} (\rho ||\sigma ) &=& M_x (\rho || \sigma) \ .
\eea
More generally, by expanding (\ref{RelativeM_first def}) with relative cumulants and setting $a_n = b_n-\ln(x)$, we have
\begin{eqnarray}\label{relativeM_variousn}
M^{(1)}_{b_1,x}(\rho || \sigma) &=& -S(\rho || \sigma)+a_1\,,  \\
M^{(2)}_{b_2,x}(\rho ||\sigma) &=& [-S(\rho || \sigma)+a_2]^2 + C(\rho || \sigma)\,,   \\
M^{(3)}_{b_3,x}(\rho || \sigma ) &=& [-S(\rho||\sigma)+a_3]^3  +3 C (\rho || \sigma ) [-S(\rho || \sigma) +a_3] - C_3 (\rho || \sigma)\,,  \\
M^{(4)}_{b_4,x}(\rho ||\sigma) &=& [-S(\rho ||\sigma ) + a_4]^4 + 6 C(\rho || \sigma)[-S(\rho||\sigma) +a_4]^2 + 3 C(\rho || \sigma)^2 \nonumber
\\
&&- 4 C_3 (\rho|| \sigma)[-S(\rho ||\sigma) + a_4]+ C_4 (\rho ||\sigma)\,,  \\
 &\cdots&\,,  \nonumber
\end{eqnarray}
where $C_n (\rho || \sigma)$ is the $n^{\rm th}$ cumulant of $\ln (\rho \sigma^{-1})$. Notice that, up to additive constants, the equations in (\ref{relativeM_variousn}) reduce to the ones in (\ref{extrM_variousn}).
% up to the substitutions $S(\rho)\to -S(\rho || \sigma)$ and $C_j(\rho)\to (-1)^jC_j(\rho || \sigma)$
This follows by choosing $\sigma =\mathbb{1}/d$ giving $S(\rho || \mathbb{1}/d)=-S(\rho)+\ln d$ and
more generally  $C_j(\rho ||\mathbb{1}/d)=(-1)^jC_j(\rho)$, and setting $a_n=n-1-\ln s_{\textrm{\tiny min}}= n-1+\ln d$.

%\textcolor{blue}{
When $\rho$ and $\sigma$ are commuting density matrices, we can also define a generating function %(TO BE CROSS-CHECKED!),
\be
k_{\rm rel}(\alpha ; a ) = e^{-\alpha a} \exp [(\alpha-1)S_\alpha (\rho || \sigma)]\,,
\ee
where
\be
S_\alpha (\rho || \sigma)=\frac{1}{\alpha -1} \ln 
\Tr \left[\rho^\alpha\sigma^{1-\alpha}\right]
%\Tr \left[\left(\sigma^{\frac{1-\alpha}{2\alpha}}\rho\sigma^{\frac{1-\alpha}{2\alpha}}\right)^\alpha\right]
\ee
is the Petz-R\'enyi relative entropy \cite{Petzbook}. In this case, we get
\be
\label{relativeM from kgen}
 M^{(n)}_{b_n,x}(\rho || \sigma)=\left[(-1)^n e^a \frac{\partial^n}{\partial \alpha^n} k_{\rm rel}(\alpha ; a)\right]\bigg |_{\alpha =1, a= b_n -\ln (x)} \ .
\ee
It would be interesting to generalize the generating function giving (\ref{relativeM from kgen}) to the case of non-commuting $\rho$ and $\sigma$. We leave this investigation for a future analysis.
%}

Likewise, from $P^{(n)}_E$ we obtain the relative quantities
\be
\label{relP_def}
  P^{(n)}_{E,x} (\rho || \sigma) = -\Tr \left[\rho F_n (\ln \rho -\ln \sigma + \ln(x))\right] = -P^{(n)}_E(x\rho\sigma^{-1})\,,
\ee
where $F_n$ is an extremal polynomial that is a solution of $F_n''(y)+F_n'(y)=G_n(y)$, where $G_n(y)$ is of the form given in Theorem 1. 
%The term subtracted in the rhs of (\ref{relP_def}) is such that $P^{(n)}_{E,x} (\sigma || \sigma)=0 $.

In Ref. \cite{boes2020} the monotonicity of $M_{s_{\rm min}}(\rho ||\sigma)$ (see \ref{relM2}) has been used to derive a lower bound for relative entropy production. As corollary of the more general theorem (\ref{Eineq}), one can obtain infinitely many inequalities involving the change in relative entropy. If  $(\rho,\sigma)\succ (\rho',\sigma')$ with both $\sigma, \sigma'$ full rank, then
\bea
M^{(n)}_{b_n,s_{\rm min}}(\rho' || \sigma') &\geq& M^{(n)}_{b_n,s_{\rm min}}(\rho || \sigma)\,, \\
  P^{(n)}_{E,s_{\rm min}} (\rho' || \sigma') &\geq&   P^{(n)}_{E,s_{\rm min}} (\rho || \sigma) \ , 
\eea
for all $n\geq 1$. Substituting the explicit form of $M^{(n)}_{b_n,s_{\rm min}}$ from (\ref{RelativeM_first def}) leads to inequalities resembling the ones in Sec. \ref{subsec:ExtremalMonotones},
from which one can solve the relative entropy production with a bound involving relative cumulants to arbitrary order. For example, at order $n=2$, concavity requires $b_2\geq 1$, leading to \cite{boes2020}
\be
\label{rel_ineq_Wilming}
\Delta S_{\rm rel}^2+2\chi \Delta S_{\rm rel} -\Delta C_{\rm rel}\geq 0 
\ee
where $\Delta S_{\rm rel} \equiv S(\rho ||\sigma)-S(\rho' ||\sigma')$, $ \Delta C_{\rm rel} \equiv C(\rho ||\sigma)-C(\rho'||\sigma')$ and $\chi=a-S(\rho ||\sigma)$, with $a\equiv1-\ln s_{\textrm{\tiny min}}$.
The inequality (\ref{rel_ineq_Wilming}) can be rewritten as
the following relative entropy production bound 
\be\label{relentproduction_new}
 \Delta S_{\rm rel} \geq \frac{\Delta C_{\rm rel}}{2 a- S(\rho ||\sigma)-S(\rho' ||\sigma')}  \,.
\ee
%(see \cite{boes2020} for the explicit calculation)
 When $\sigma$ and $\sigma'$ are the maximally mixed state, (\ref{relentproduction_new}) becomes (\ref{inequality_new n2}).
The bound in (\ref{relentproduction_new}) can be relaxed obtaining the somewhat less tight inequality (see \cite{boes2020} for the explicit calculation)
\be\label{relentproduction}
 \Delta S_{\rm rel} \geq \frac{\Delta C_{\rm rel}}{2\sqrt{M^{(2)}_{1,s_{\rm min}}(\rho || \sigma)}} \geq \frac{\Delta C_{\rm rel}}{2(1-\ln{s_{\rm min}})} \,.
\ee
 The inequality (\ref{relentproduction}) is the same as that reported in \cite{boes2020} with
$1-\ln{s_{\rm min}}$ replacing $1/\ln 2 - \log_2(s_{\rm min})$  due to the use of  $\ln$ instead of $\log_2$. Notice that the inequality (\ref{Landauer-ineq}) can be recovered from (\ref{relentproduction}) when $\sigma$ and $\sigma'$ are the maximally mixed state. 

Moving to the $P^{(n)}_{E,s_{\rm min}}$, tighter inequalites can be obtained at orders $n\geq 3$. For example, at
order $n=3$,  writing
\be
0 \leq P^{(3)}_{E,s_{\rm min}}(\rho' ||\sigma' ) - P^{(3)}_{E,s_{\rm min}} (\rho ||\sigma) \equiv w_2 a^2 + w_1 a + w_0\,,
\ee
where (using $b_n=n-1$)
\bea
&& w_2 = M^{(2)}(s_{\rm min}\rho'\sigma'^{-1};1)-M^{(2)}(s_{\rm min}\rho\sigma^{-1};1) = M ^{(2)}_{1,s_{\rm min}}(\rho' ||\sigma') - M^{(2)}_{1,s_{\rm min}}(\rho ||\sigma)\,, \\
&& w_1 =  -M ^{(1)}_{0,s_{\rm min}}(\rho' ||\sigma') + M^{(1)}_{0,s_{\rm min}}(\rho ||\sigma)\,, \\
&& w_0 =\frac{1}{3}\big[ M ^{(3)}_{2,s_{\rm min}}(\rho' ||\sigma') - M^{(3)}_{2,s_{\rm min}}(\rho ||\sigma)\big] -\big[M ^{(2)}_{1,s_{\rm min}}(\rho' ||\sigma') - M^{(2)}_{1,s_{\rm min}}(\rho ||\sigma)\big] \,,
\eea
 we obtain 
\be\label{Relineq3compact}
\Delta M^{\rm rel}_3 \geq 3 \Delta M^{\rm rel}_2 + \frac{3}{4} \frac{(\Delta M^{\rm rel}_2)^2}{\Delta M^{\rm rel}_1} \,,
\ee
where now
\be\label{deltaM_extremal}
\Delta M^{\rm rel}_n \equiv  M^{(n)}_{n-1,s_{\rm min}} (\rho' || \sigma') -  M^{(n)}_{n-1,s_{\rm min}} (\rho || \sigma) \ .
\ee

\subsubsection{Finite-size correction to Clausius inequality}
%\textcolor{red}{Comment on what finite size means: it could be confused with finite-size in many-body system and CFT.}

Rather than moving on to derive more explicit forms of these higher order inequalities, let us point out an interesting corollary of (\ref{relentproduction}). A well-known statement is that the non-negativity of relative entropy $S(\rho ||\sigma)\geq 0$ can be rewritten as an inequality 
\be\label{Kineq}
 \Delta \langle K \rangle \equiv \Tr (\rho K) - \Tr (\sigma K) \geq S(\rho) - S (\sigma) \equiv \Delta S\,,
\ee
where\footnote{The inequality (\ref{Kineq}) is also invariant under a change in the normalization of $\sigma$, obtained through the shift of $K$ by the proper constant.} $K = -\ln \sigma$. The refined equation (\ref{relentproduction}) gives a finite-size correction to the inequality (\ref{Kineq}), assuming $[\rho,\sigma]=0$. 
Set $\rho'=\sigma = \sigma'$ in (\ref{relentproduction}), which corresponds to consider a mutually commuting pair of states with the $\sigma$-majorization $\rho \succ_\sigma \sigma$. Now (\ref{relentproduction}) reduces to 
\be
S(\rho ||\sigma) \geq \frac{C(\rho ||\sigma )}{2\sqrt{M^{(2)}_{1,s_{\rm min}}(\rho || \sigma)}} \geq \frac{C(\rho ||\sigma )}{2(1-\ln s_{\rm min})}\,,
\ee
which in turn implies a finite-size correction to (\ref{Kineq}); indeed
\be\label{Kineq2}
  \Tr (\rho K) - \Tr (\sigma K) \geq S(\rho) - S (\sigma) + \frac{C(\rho ||\sigma )}{2\sqrt{M^{(2)}_{1,s_{\rm min}}(\rho || \sigma)}} \ .
\ee
In the context of non-equilibrium thermodynamics, we can choose $\sigma = \gamma_\beta \equiv e^{-\beta H}/Z(\beta)$, where $\beta = 1/(k_B T)$ is the inverse temperature, and $\rho$ is a classical state that thermomajorizes $\gamma_\beta$.
From (\ref{Kineq2}), we obtain  a finite-size correction to a fundamental thermodynamic relation (the Clausius inequality):
\begin{eqnarray}\label{Kineqthermo}
 S (\gamma_\beta )-  S (\rho) &\geq &  \frac{1}{k_BT}[\langle H\rangle_{\gamma_\beta} - \langle H \rangle_\rho]  + \frac{C(\rho ||\gamma_\beta)}{2\sqrt{M^{(2)}_{1,s_{\rm min}}(\rho || \gamma_\beta)}} \\
                                  &\geq &\frac{1}{k_BT}[\langle H\rangle_{\gamma_\beta} - \langle H \rangle_\rho] +\frac{C(\rho ||\gamma_\beta )}{2+ 2\beta (E_{\rm max} -F(\beta))}\ , \nonumber
\end{eqnarray}
where  $\langle H\rangle_\rho $ is the non-equilibrium energy of the state $\rho$, $\langle H\rangle_{\gamma_\beta}$ the energy in the thermal state, $S(\rho)$ the non-equilibrium entropy in the state $\rho$, $S(\gamma_\beta)$ the thermal entropy
(or, more conventionally, $S(\beta)=k_BS(\gamma_\beta)$),  and where $s_{\rm min}$ is the smallest eigenvalue $\exp (-\beta E_{\rm max})/Z(\beta)$ of the Gaussian state $\gamma_\beta$ (the finite dimensional system has a maximum energy eigenvalue $E_{\rm max}$) and $\beta F(\beta)=-\ln Z(\beta)$ is the Helmholtz free energy.
 Again, one could attempt to derive a sequence of more refined inequalities from the infinite sequences of resource monotones above constructed.

Let us remark that the finite-size correction obtained above has no effect on the so-called first law of entanglement. Consider a one-parameter family of states $\rho (\lambda)$, with $\rho (0) =\sigma$. Expanding $\rho (\lambda ) = \sigma + \lambda \delta \rho + \cdots$,
one obtains
\be
\delta S(\rho ||\sigma) \equiv S(\rho(\lambda)||\sigma) - 0 = 0 + \lambda S^{(1)}(\sigma ||\sigma) + \lambda^2 S^{(2)}(\sigma ||\sigma) + \cdots \geq 0 \,. 
\ee
Since $S(\rho ||\sigma)$ has a global minimum at $\rho = \sigma$, at first order we have an equality $\delta S(\rho || \sigma)= \lambda \,S^{(1)}(\sigma ||\sigma )=0$ instead of an inequality, hence (\ref{Kineq}) becomes the equality
\be
  \delta \langle K \rangle = \delta S\,,
\ee
for infinitesimal changes, which is known as the first law of entanglement. The reason why it does not receive finite-size corrections from (\ref{Kineq2}) is that $\rho =\sigma$ is a local minimum of $C(\rho ||\sigma)$, so $\delta C(\rho ||\sigma) =0$ also to first
order, so that the first order equality $\delta S(\rho ||\sigma)=0$ remains intact.

\section{Applications}
\label{sec:applications}

\subsection{Information erasure}
\label{subsec:InfoErasure}

%\textcolor{red}{All over this section: \textcolor{red}{$I$ or $\mathbb{1}$?}}

As an application of the inequalities (\ref{j2}) and (\ref{M3inequality}), we study a state erasure process, following \cite{boes2020}. The erasure process is described in two steps. 
Consider first a system in a state $\rho$ to be "erased", {\em i.e.} to be converted to some fixed pure state $|\psi \rangle$. This conversion can be accomplished introducing
an external system $\mathcal{B}$ involving $n$ qubits, an {\em information battery}, which is simultaneously converted from a pure state $|0\rangle^{\otimes n}$ to a maximally mixed flat 
state $(\mathbb{1}/2)^{\otimes n}$. This first step of the process requires the battery to be large enough, meaning that there is a lower bound on the required $n$. Mathematically, this step can be modeled by a unital channel $\Ecal$ 
mapping the input state $\Omega$  of the composite system  to the output state $\Upsilon=\Ecal (\Omega)$. According to the Uhlmann's theorem discussed in Sec.\,\ref{subsec:concavity},  the unital channel implies the majorization relation 
\be
\Omega \equiv \rho \otimes |0\rangle \langle 0|^{\otimes n} 
\succ 
|\psi \rangle \langle \psi | \otimes (\mathbb{1}/2)^{\otimes n} \equiv \Upsilon \, .
\ee  
From the Schur concavity of von Neumann entropy it follows that $S(\Omega) \leq S(\Upsilon)$, which by additivity of $S$ produces the lower bound 
\be\label{firstbound}
n\ln 2 \geq S(\rho )\,.
\ee

In the second step, in order to be able to repeat the procedure, one must restore the battery back to the initial state $|0\rangle^{\otimes n}$. This step requires work to be done on the battery, and can be performed in different ways (see {\em e.g.} \cite{faistthesis} for more discussion and references). We consider
here the process described in \cite{dahlsten}, where the battery consists of $n$ identical copies of a tunable two-state system and the energy gap of the two energy eigenstates can be parametrically adjusted to be anything from zero towards infinity. One
places the battery in contact with a heat bath at temperature $T$,  with the energy gap initially being zero, and then adiabatically increases the gap towards infinity. This requires work $W$ to be performed on the battery. 
At the end, each two-state system is in the ground state $\ket{0}$ with probability 1,
so that the battery is restored to the state $\ket{0}^{\otimes n}$, at the expense of the work cost involved in the adiabatic process, which can be calculated to be  
\be\label{work}
W =k_BTn\ln 2\ .
\ee  
One then removes the battery from the contact with the heat bath, after which the
energy gap can be reduced back to zero while doing no work. The battery is now ready to be used again.  

The two equations (\ref{firstbound}) and (\ref{work}) can be combined as the Landauer inequality 
\be
  W \geq k_B T S(\rho) \,, 
\ee
for the work cost of erasure. Note that erasure deletes all details of the eigenvalue spectrum of $\rho$, so it is
natural to ask if additional details of the distribution of the spectrum  beyond just the entropy will have an effect on the work cost of erasure. Indeed we find that all cumulants of the spectrum can be used to characterize
the cost.

The inequality (\ref{firstbound}) has been derived from $S$, but now we may consider the infinite sequence inequalities following from the monotones $M^{(m)}$.
% Let us consider that and compare with the previous method. %Let us take the parameter $b_n$ to saturate the lower bound, $b_n=n-1$. 
Since $M^{(m)}$ are (Schur) concave, we have that $\Omega \succ \Upsilon \Rightarrow M^{(m)}(\Upsilon ,b_m) \geq M^{(m)}(\Omega ,b_m)$. First we compute 
\be
M^{(m)}(\Upsilon ,b_m ) = \sum^{m}_{k=1}  {m \choose k} b^{m-k}_m \mu_k (\Upsilon ) \,,
\;\;\;\qquad\;\;\;
\mu_k (\Upsilon ) = \sum'_{p_j}k!\prod_j \frac{1}{p_j!(j!)^{p_j}} C^{p_j}_j (\Upsilon ) \ ,
\ee
where $\sum'$ is the restricted sum introduced in (\ref{momentstocumulants}) and, by additivity of cumulants (\ref{additivity_Cn}) and the fact that they are zero for pure states, we have
\be
C_j (\Upsilon ) = 0 + n\,C_j (\mathbb{1}/2) \ .
\ee
Since $C_j (\mathbb{1}/2)=0$ for $j\geq2$, in the above sum only the partition $k=1+1+\cdots +1$ remains, 
hence%
\be
M^{(m)} (\Upsilon ,b_m ) = \sum^m_{k=0}{m \choose k} b^{m-k}_m (n\ln 2)^k = (n\ln 2+b_m)^m \ .
\ee
Next, we have
\begin{eqnarray}
M^{(m)}(\Omega ,b_m) &=& \sum^{m}_{k=0}  {m \choose k} b^{m-k}_m  \sum'_{p_j}k!
\prod_j \frac{1}{p_j!(j!)^{p_j}} \;C^{p_j}_j (\Omega )\,,
\end{eqnarray}
where (using again the additivity (\ref{additivity_Cn}))
\be
C_j (\Omega ) = C_j (\rho ) + 0 \ .
\ee
We thus have confirmed that
\be
M^{(m)} (\Omega ,b_m) = M^{(m)} (\rho ,b_m) \ .
\ee
Now we employ the cumulant expansion (\ref{momentstocumulants2}).
%We then expand these moments of $-\log \rho + b_m$ as cumulants of $-\log \rho + b_m$ . In the sum over partitions, terms involving cumulants of $-\log\rho + b_m$
% of order $j\geq 2$ reduce to $C_j (\rho)$ due to translation invariance. Terms involving $j=1$ yield powers of $S(\rho) + b_m$. 
Setting for simplicity $b_m$ to be
the smallest possible value $b_m=m-1$ allowed by concavity, we obtain the sequence of inequalities for the work cost $W= k_BTn\ln 2$ with\footnote{There is again a subtlety with conventions: we get $S((\mathbb{1}/2)^{\otimes n}) = n \ln 2$, whereas (\ref{Msequence1}) as quoted in \cite{boes2020} involves $n$, due the convention of using binary logarithms producing
$\tilde S((\mathbb{1}/2)^{\otimes n})=n$. }  
\begin{eqnarray}\label{Msequence1}
 n\ln 2 &\geq& S(\rho)\,, \\
\label{Msequence2}
 (n\ln 2+1)^2 &\geq& \big[S(\rho)+1\big]^2 + C(\rho ) \,,\\
\label{Msequence3}
 (n\ln 2+2)^3 &\geq& \big[S(\rho)+2 \big]^3  + 3 C (\rho) \big[S(\rho) +2\big] + C_3 (\rho) \,,\\
 (n\ln 2+3)^4 &\geq& \big[S(\rho) + 3\big]^4 
 + 6 C(\rho )\, \big[S(\rho) +3\big]^2 + 3 C(\rho)^2 + 4 C_3 (\rho)\,\big[S(\rho) + 3\big] + C_4 (\rho)\,, 
 \hspace{1cm} \\
 &\cdots&\,.  \nonumber
\end{eqnarray} 

Let us remark on the relative significance of the various terms on the right hand side of these inequalities. We point out that, while $S_{\max} =\ln d$ by the maximally mixed state $\rho_{\rm flat}=\frac{1}{d}\mathbb{1}$,
we have $C_{\max }\approx \frac{1}{4}\ln^2 (d-1)$ \cite{Reeb_2014, deBoer:2018mzv, boes2020} by a different state $\rho_{\max} = \diag (1-r,\frac{r}{d-1},\ldots, \frac{r}{d-1} )$, where $r$ is the solution of $(1-2r)\ln [(1-r)(d-1)/r]=2$.
One might anticipate that the higher cumulants have maximum values $C_{n,\max}\propto \ln^n d$ each reached at a different state $\rho_{n,\max}$, but this remains (to our knowledge) an open problem.

Let us now compare these inequalities with those ones obtained from the extremal polynomial monotones. The first two are the same as in the above:  $P^{(1)}_E = -S=-M^{(1)}$ gives the Landauer bound (\ref{Msequence1}),
%\be\label{Landauer}
while 
%\ee 
the inequality from $P^{(2)}_E = -M^{(2)}$ will be the same as (\ref{Msequence2}).  Let us verify the latter inequality starting from (\ref{j2}), which first gives $ n\ln 2 \geq S(\rho ) + \frac{C(\rho)}{S(\rho)+n\ln 2 +2}$, or
%\be
% n \geq S(\rho ) + \frac{C(\rho)}{S(\rho)+n +2} \,,
%\ee
%or
\be
\label{Landauer-finite-v2}
 n\ln 2+1 \geq \sqrt{(S(\rho ) +1)^2 + C(\rho)} \, ,
\ee
which coincides with (\ref{Msequence2}). 
This is already a stronger bound than the one coming from the inequality (\ref{Landauer-ineq}) involving $S$, $C$ and $M$,
which reads 
\be\label{Landauer-finite}
 n\ln 2\geq S(\rho) + \frac{C(\rho )}{2\sqrt{M(\rho )}} \, .
\ee 

Consider then the tighter third order inequality (\ref{M3inequality}) applied to the erasure process. 
We need the moments of modular Hamiltonian $\mu_n$ (see (\ref{lincombmoments})) for the states 
$\Omega$ and $\Upsilon$ in terms of the cumulants. 
From (\ref{momentstocumulants}), for $\Omega= \rho \otimes |0\rangle \langle 0|^{\otimes n}$ we have
\bea
&& \mu_1(\Omega) = S(\rho)\,,  \\
&& \mu_2 (\Omega) = S^2 (\rho) + C(\rho )\,,  \\
&& \mu_3(\Omega) 
= S^3 (\Omega) + 3S(\Omega)\,C(\Omega )+ C_3(\Omega) 
= S^3(\rho) + 3S(\rho)\,C(\rho) + C_3(\rho)\,,
\eea
and for $\Upsilon=|\psi \rangle \langle \psi | \otimes (\mathbb{1}/2)^{\otimes n}$ we obtain 
$\mu_k(\Upsilon)=(n\ln 2)^k$. %(with factors of $\ln 2$ to be absorbed into the work cost unit).
After some calculation, we find 
\bea
\label{Landauer3}
&&(n\ln 2+1)^3 - [S(\rho )+1]^3 + 3 [n\ln 2-S(\rho )] - [C_3(\rho ) + 3C(\rho)(S(\rho )+1)] \nonumber 
\\
\rule{0pt}{.7cm}
&& 
\geq \frac{3}{4}\frac{[(n\ln 2+1)^2 -(S(\rho) +1)^2 -C(\rho )]^2}{(n\ln 2-S(\rho))} \ .
\eea
From (\ref{Ineq3compact}), we already know that this is a stronger inequality than (\ref{Msequence3}).
%Note that in addition to the variance, the inequality now contains also the third cumulant $C_3(\rho )$. 
A weaker form of the inequality (\ref{Ineq3compact}) is $\Delta M_3 \geq 3 \Delta M_2$ and it gives
\be\label{Landauer3weak}
 (n\ln 2+1)^3 + 3 n\ln 2 \geq (S(\rho )+1)^3 + 3 S(\rho )  + 3C(\rho)(S(\rho )+1) +C_3(\rho ) \ ,
\ee
which is still somewhat stronger than (\ref{Msequence3}).

Moving to higher order extremal polynomial monotones, one can derive an infinite sequence of (increasingly more complicated) inequalities, involving higher cumulants of modular Hamiltonian. 
As already noted, it is possible that the sequence contains some interesting hierarchical structure, {\em e.g.} the right hand side of (\ref{Landauer3}) contains a ratio of quantities that appears in 
the two previous inequalities (\ref{Msequence2}) and (\ref{Msequence1}).

In the outlook section of \cite{boes2020}, it was commented that it would be interesting to construct a hierarchy of (Schur-)concave functions from cumulants of modular Hamiltonian relevant to single-shot settings. We have proposed two such sequences: the moments of shifted modular Hamiltonian and the extremal polynomial monotones. The latter one could be relevant to provide tight inequalities. It would be interesting to further explore the hierarchical structrues contained in the two sequences. Furthermore, besides a "sequence of Landauer inequalities", the majorization relation in the information erasure already contains many inequalities for the partial sums of the two eigenvalue spectra. Since any concave quantifier
produces an inequality for the information erasure, and thus some kind of a bound on $n$, there is obviously an uncountable infinity of inequalities that one could consider. The above sequences are special only in the sense that they contain the entropy
and higher cumulants.

Majorization gives a partial ordering among density matrices, but one can also define another partial ordering based on the sequences of inequalities from the monotones. It would be interesting to study the relation between these two orderings in more detail and some tentative observations are made in Sec.\,\ref{sec:conclusions}.

\subsubsection{Lower bound on marginal entropy production}

The inequality (\ref{j2}) offers a slight tightening of the lower bound on marginal entropy production derived in \cite{boes2020}. Our discussion follows again \cite{boes2020} but with some small modifications. Consider a quantum channel $\Ecal$ mapping a system $\mathrm{S}$ to itself. In order to represent the channel $\Ecal$ acting on a state $\rho_{\rm S}$ of the system, we introduce an environment $\mathrm{E}$ in a state $\rho_{\rm E}$ and a {\em unital} channel\footnote{By the Stinespring dilation
one could choose a unitary channel, but we follow \cite{boes2020} and consider this more general possibility.} $\Ucal$ to write
\be
 \Ecal (\rho_{\rm S}) = \mathrm{Tr}_{\rm E} [\Ucal (\rho_{\rm S} \otimes \rho_{\rm E})] \ .
\ee
 We denote by 
\be
\rho'_{\rm SE} = \Ucal (\rho_{\rm S}\otimes \rho_{\rm E})
\ee
the state of the joint system $\mathrm{S}\cup\mathrm{E}$ after the application of the unital channel. Since $\Ucal$ is unital, by the Uhlmann's theorem mentioned in Sec.\,\ref{subsec:concavity}, we have the majorization $\rho_{\rm E} \otimes \rho_{\rm S} \succ \rho'_{\rm SE}$ and therefore, using (\ref{j2}), we obtain
\be\label{step1}
(S(\rho'_{\rm SE})+1)^2 - (S(\rho_{\rm S}\otimes \rho_{\rm E})+1)^2 \geq C(\rho_{\rm S} \otimes \rho_{\rm E}) - C(\rho'_{\rm SE}) \,,
\ee
which means that the entropy production in the joint system  is bounded from below by the decrease of the variance. 
The right hand side can be manipulated using a result (Lemma 11) from \cite{boes2020} which takes the form of a 
correction to the subadditivity of variance. To describe that result we consider 
a bipartite system $\mathrm{S}\cup\mathrm{E}$ of dimension $d=d_{\mathrm{S}}d_{\mathrm{E}}\geq 2$, where $\rho_\mathrm{SE}$ and $ \sigma_\mathrm{SE}$ are two commuting quantum states. If $\sigma_\mathrm{SE} = \sigma_{\mathrm{S}} \otimes \sigma_\mathrm{E}$ and is full rank,
then\footnote{There are again some rescalings by factors of $\ln 2$, due to our convention of using $\ln$ instead of $\log_2$.}
\be
C(\rho_\mathrm{SE} ||\sigma_\mathrm{SE}) \leq C(\rho_\mathrm{S} || \sigma_\mathrm{S}) + C(\rho_\mathrm{E} || \sigma_\mathrm{E} ) + \kappa \, f (I_{\mathrm{SE}}/\ln 2)\,,
\ee
where $\rho_\mathrm{S} = \Tr_E \rho_\mathrm{SE},\, \rho_\mathrm{E} = \Tr_\mathrm{S}\rho_\mathrm{SE}$, $\kappa$ is a constant given by
\be
\kappa = \sqrt{2\ln 2}\,\big(12\ln^2 2 + \ln^2 s_{\rm min} + 8\ln^2 d\big)\,,
\ee
where $s_{\rm min}$ is the smallest eigenvalue of $\sigma$,  $I_{{\mathrm{SE}}}$ is the mutual information
\be
 I_{\mathrm{SE}}= S(\rho_\mathrm{SE} || \rho_\mathrm{S} \otimes \rho_\mathrm{E})
\ee
of $\rho_{\rm SE}$ with respect to the bipartition and $f(x)=\max \{x^{1/4},x^{1/2}\}$.  As a special case, choosing $\sigma_\mathrm{SE}=(\mathbb{1}_\mathrm{S}/d_\mathrm{S})\otimes (\mathbb{1}_\mathrm{E}/d_\mathrm{E})$ with $s_{\rm min}= \frac{1}{d_\mathrm{S}d_\mathrm{E}}=\frac{1}{d}$, we obtain the following correction to the subadditivity of variance
\be\label{Csubadd}
C(\rho_{\rm SE}) \leq C(\rho_{\rm S}) + C(\rho_{\rm E}) + \kappa  f(I_{\rm SE}/\ln 2) \ ,
\ee
where now 
\be
\kappa = \sqrt{2\ln 2}\,\big(12\ln^2 2 + 9 \ln^2 d\big)\,.
\ee
We can now turn back to (\ref{step1}) and bound the right hand side that equation using the additivity of the variance (\ref{additivity_Cn}), $C(\rho_\mathrm{S}\otimes \rho_\mathrm{E})=C(\rho_\mathrm{S})+C(\rho_\mathrm{E})$ and the correction to the subadditivity (\ref{Csubadd}) to decompose $C(\rho'_{\rm SE})$, ending up with 
\be
 C(\rho_{\rm S} \otimes \rho_{\rm E}) - C(\rho'_{\rm SE})  \geq -\Delta C_{\rm S} -\Delta C_{\rm E} - \kappa\,  f(I_{\rm SE}/\ln 2)  \ ,
\ee
where we have defined $\Delta C_{\rm S} = C(\rho'_{\rm S})-C(\rho_{\rm S})$ and $\Delta C_{\rm E} = C(\rho'_{\rm E})-C(\rho_{\rm E})$. Note that, for decreasing variance, $-\Delta C\geq 0$. In the left hand side of (\ref{step1}) we can use the additivity and subadditivity of entropy to write
\be
(S_{\rm S,E} + \Delta S_{\rm S,E}+1)^2 - (S_{\rm S,E}+1)^2 \geq (S(\rho'_{\rm SE})+1)^2 - (S(\rho_{\rm S}\otimes \rho_{\rm E})+1)^2\,,
\ee
where $S_{\rm S,E}\equiv S(\rho_{\rm S})+ S(\rho_{\rm E})$ and $\Delta S_{\rm S,E}\equiv\Delta S_{\rm S} + \Delta S_{\rm E} \equiv S(\rho'_{\rm S})-S(\rho_{\rm S})+ S(\rho'_{\rm E})-S(\rho_{\rm E})$. Combining everything together, and expressing the resulting inequality for the
entropy production, we have 
\be
\label{bound_marginalent_new}
\Delta S_{\rm S} + \Delta S_{\rm E} \geq (S(\rho_{\rm S}) + S(\rho_{\rm E})+1)\left\{ \left[ 1 + \frac{-\Delta C_{\rm S} -\Delta C_{\rm E} - \kappa \, f(I_{\rm SE}/\ln 2)  }{ (S(\rho_{\rm S}) + S(\rho_{\rm E})+1)^2}\right]^{1/2}-1 \right\} \, ,
\ee
which is a slightly tighter lower bound for marginal entropy production, compared to Result 4 in \cite{boes2020}. We emphasize that the result is non-trivial only when the numerator in the ratio in the square brackets is positive (the total decrease in 
variances is sufficiently large). If the ratio in the square bracket in the right hand side of (\ref{bound_marginalent_new}) is much smaller than one, we can employ the Taylor expansion, obtaining from the leading approximation
\be
\Delta S_{\rm S} + \Delta S_{\rm E} \gtrsim \frac{-\Delta C_{\rm S} -\Delta C_{\rm E} - \kappa \, f(I_{\rm SE}/\ln 2)  }{ 2(S(\rho_{\rm S}) + S(\rho_{\rm E})+1)},
\ee
which can be easily compared against the inequality in \cite{boes2020}\footnote{Note that we have opposite sign conventions for $\Delta S$, $\Delta C$.}. 
As an application, \cite{boes2020} considered state transitions with a help of a catalytic system $\mathrm{C}$ and derived the lower bound $d_\mathrm{C} \geq O[\exp (\delta^{-1/8})]$ for the necessary dimension $d_\mathrm{C}$ of  $\mathrm{C}$ for a state transition, where the variation of entropy $\delta$ is small while variance is reduced. 
Similarly, it would be interesting to consider what are the implications of the other inequalities in the sequences considered in this manuscript, which involve higher cumulants. For that,
one would need to first derive subadditivity properties of higher cumulants $C_n$, generalizing the result (\ref{Csubadd}) for the variance.

\subsection{Majorization and states of periodic chains
%and QFT states
}
\label{subsec:thougthexp}

Majorization is a central concept in finite dimensional quantum systems, which provides a classification of bipartite entanglement for pure states. What can be said about pure states in quantum field theories? Consider a pair of pure states
 $\ket{\Phi}$ and $\ket{\Psi}$ in a QFT, is there any meaningful definition for a relation $\ket{\Phi} \succ \ket{\Psi}$? First of all, bipartitioning involves an
ultraviolet cutoff, making some entanglement monotones divergent. Also, the definition of majorization via the reduced states and partial sums of ordered eigenvalues becomes cumbersome and very sensitive to the choice of UV regulator. In Sec.\,\ref{CFT} we will explore this question to gain some tentative insight. We consider a pair of states in a CFT which is a continuum limit of a discrete model at criticality. As a concrete example, we will consider discrete versions of the compact boson at critical radius and free fermion CFT on a circle, which are different critical limits of the periodic anisotropic XY spin chain. The latter can be mapped to free periodic fermionic chains
with a finite dimensional Hilbert space. The original pair of CFT states is then mapped to their counterpart in the fermionic chain.  In that case majorization becomes well defined, and  we can apply rigorous theorems from quantum information theory and
compare bipartite entanglement between the two states. The results will then reflect some properties of the original CFT states, which have to be interpreted with care, as the maps between the spin chains and fermionic chains are typically non-local.

Our goal is to take some first steps, leaving more exhaustive studies for future work.  We will limit to comparing a ground state $\ket{\Phi}\equiv \ket{0}$ with an excited state $\ket{\Psi}$ and their counterparts in the periodic fermionic chain.
%\textcolor{red}{
%We divide the periodic chain of finite length $L$ (with a finite number $N$ of sites) to a subregion of length $\ell$ (with a finite number of sites $M$). The
%full system Hilbert space then will have dimension $2^N$  while the subregion is a subsystem of dimension $2^M$. 
%}
We bipartite the periodic chain of finite length $L$ made by $N$ sites 
(hence the Hilbert space of the full system has dimension $2^N$)
into a subregion made by $\ell$ consecutive sites and its complement. Let us denote this subsystem by $A$ and its complement as $B$.
We denote the reduced states in the subsystem $A$ by
%\textcolor{teal}{\bf Every time we have $\rho$ and $\sigma$ in this section we propose to use the subscript $A$ since they are always reduced density matrices.}
\be
\rho_A \equiv \Tr_B (\ket{\Phi} \!\bra{\Phi} ) \,, 
\;\;\;\qquad\;\;\;  
\sigma_A \equiv \Tr_B (\ket{\Psi}\!\bra{\Psi}) \ .
\ee
Suppose that $\ket{\Phi} \succ \ket{\Psi}$  or equivalently $\rho_A \succ \sigma_A$. Note that this relation depends on the choice of the bipartition to $A$ and $B$, hence on the relative size of the two lengths $\ell /L$. While it is
rather laborius to establish majorization, we ask an easier question of whether it can be ruled out. For this purpose we can use any Schur concave quantity $E_A$ applied to the reduced
states:
%, such as the entanglement entropy, R\'enyi entropies, or the monotones $M^{(n)},P^{(n)}_E$, 
assuming that $\rho_A \succ \sigma_A$  ($\ket{\Phi} \succ \ket{\Psi}$) then $\Delta E_A\equiv E_A (\sigma_A) -E_A (\rho_A) \geq 0$, thus if we instead find $\Delta E_A \leq 0$ the  
%(or $-\Delta E_A\equiv E_A (\rho_A) -E_A (\sigma_A) \geq 0$). If we then find $\Delta E_A <0$ for any Schur concave
assumed majorization order cannot hold. Consequently any process converting one state into the other implying the assumed order is impossible (for example a LOCC process 
$\ket{\Psi} \underset{\textrm{\tiny LOCC}}{\longrightarrow} \ket{\Phi}$). The converse analysis naturally can be applied to the opposite assumption $\sigma_A \succ \rho_A$ and opposite processes. 
In this setting it becomes interesting to compare which quantity $E_A$ gives the most stringent bound.

For the Schur concave quantities $E_A$ to study, we only consider here the entanglement entropy $S_A$, the R\'enyi entropies $S^{(2)}_A,\,S^{(3)}_A$ and the new monotone $M^{(2)}_A(\cdot ,1)=-P^{(2)}_A (\cdot )$. The latter gives the inequality (\ref{j2}), which must be statified
if majorization holds. 
%\textcolor{blue}{We will then also compare against the previously found inequality (), and see that the former gives a tighter bound.
%{\bf If we remove the figure with the check of these the inequalities, we can remove this sentence as well.}
%}
 All in all, in studying the various monotones, we will compare which one gives the tightest bound for the
range of $\ell /L$ where the majorization order must be violated. These results are discussed in detail in Sec.\,\ref{subsec:thoughtexpCFT}. 
%for some of these inequalities in the context of periodic fermionic chains (and their underlying field theories), whether a transition is ruled out or not depends on the relative size $\ell /L$ of the subsystem $A$.
Before such comparisons, we will examine the cutoff dependence of $S_A,C_A$ and $M^{(2)}_A$. We will also compare the entropy $S_A$ with the capacity $C_A$ for the two 
states, and examine how they depend on the relative size $\ell /L$.

\section{Pairs of bipartite pure states in 1+1 CFT and periodic chains}
\label{CFT}

In this section we first evaluate $S_A$ and $C_A$ in one dimensional 
and  translation invariant  systems. We review known  results for R\'enyi entropies  for a ground state and a class of excited states. From the R\'enyi entropies we first compute the difference $S_A-C_A$ and study how it changes moving from 
the ground state to an excited state. We then compute monotones and inequalities, as outlined in Sec.\,\ref{subsec:thougthexp}, and find the constraints they present for the majorization order between a ground state and an excited state.

\subsection{Some known CFT results}

In this section we study $S$, $C$ and $M$ (see (\ref{defcapacity}) and (\ref{Mdefinition})) in simple bosonic, fermionic conformal field theories and related discrete models (which can be mapped to fermionic chains).   Considering a 2D CFT
for certain states (including the ground state) and when the subsystem $A$ is a single interval of length $\ell$,  it has been found that \cite{Callan:1994py,Holzhey:1994we,Calabrese:2004eu,Cardy:2016fqc} 
%In CFT the quantity $\textrm{Tr}\hat{\rho}_A^n$ can be computed as the partition function on a $n$-sheeted Riemann surface. In \cite{Cardy:2016fqc} it is shown that, in the cases in which the geometry of the Riemann surface can be mapped to the one of an annulus, the partition function reads
\be
\label{Trrhon_CFT}
\textrm{Tr}\rho_A^n
\,=\,
c_n \;
e^{-\frac{c}{12}\big(n-\frac{1}{n}\big)W_A}+\dots\,,
\ee
where $c$ is the central charge of the CFT and
$W_A$ is a function of $\ell$ which diverges as the UV cutoff $\epsilon\to 0$ and depend also on the state and on the geometry of the entire system.
The constant $c_n$ is model and boundary condition dependent and $c_1 = 1$ 
(because of the normalisation of $\rho_A$).
For instance, when the system is on the infinite line and in its ground state,
when the system is on the circle of length $L$ and in its ground state
or when the system is on the infinite line and at finite temperature $1/\beta$,
for $W_A$ we have respectively
\be
\label{W_A CFT-known}
W_A=2\ln\!\bigg(\frac{\ell}{\epsilon}\bigg)\,,
\;\;\qquad\;\;
W_A=2\ln\!\bigg(\frac{L}{\pi\epsilon}\sin\frac{\pi \ell }{L}\bigg)\,,
\;\;\qquad\;\;
W_A=2\ln\!\bigg(\frac{\beta}{\pi\epsilon}\sinh\frac{\pi \ell}{\beta}\bigg)\,.
\ee
%when $A$ is an interval in a finite system with length $L$ descrbed by a CFT in the ground state in which one obtains
%\be
%\label{W finite size}
%W_A=2\log\bigg(\frac{L}{\pi\epsilon}\sin\frac{\ell \pi}{L}\bigg)
%\ee
%and when $A$ is an interval in a thermal state in an infinite line at temperature $\beta^{-1}$ and one has
%\be
%\label{W finite T}
%W_A=2\log\bigg(\frac{\beta}{\pi\epsilon}\sinh\frac{\ell \pi}{\beta}\bigg)
%\ee
%The quantity $\epsilon$ appering in (\ref{W interval}), (\ref{W finite size}) and (\ref{W finite T}) is an UV cutoff.

By employing (\ref{Trrhon_CFT})  into the definitions (\ref{defentropy}) and (\ref{defcapacity}),
it is straightforward to find that  \cite{deBoer:2018mzv}
\be
\label{capacityequalentropy}
C_A = S_A = \frac{c}{6}\,W_A + O(1)\,,
\ee 
%\textcolor{orange}{[this is recalled many times in maybe misleading ways (there are two equalities here)]}
%\\
and that $C_A$ and $S_A$ 
differ at the subleading order $O(1)$ determined by the non-universal constant 
$c_n$\footnote{In higher dimensional CFTs and more general quantum field theories, 
the relation is more ambiguous; indeed the UV cutoff in the two quantities appears in a power law
and the quantities become more dependent on the regularization scheme (see \cite{deBoer:2018mzv} for more discussion).}.
%We explore various cases where $S_A - C_A$ is UV finite and non trivial.

%\textcolor{red}{
%We first study what happens to the equality (\ref{capacityequalentropy}) when the system is in an  excited state, instead of the ground state. We consider CFTs in a circle of length
%$L$, with the subsystem interval of length $\ell$. 
%We will find that $S_A$ and $C_A$ obtain different finite contributions depending on the ratio $\ell / L$, which breaks the equality  (\ref{capacityequalentropy}).
%% \textcolor{teal}{In this case we find expressions for the constant terms in $\ell / \epsilon$ both for $S_A$ and $C_A$, showing explicitly that they are different in a non-trivial way.} 
%We will plot the difference $S_A-C_A$ for various cases.  
%} 

In the following we report the expression of $ M^{(n)} (\rho;b_n )$ defined in (\ref{Mn Tr Fn}) for a CFT on the line in its ground state and  
an interval $A$ of length $\ell$. 
By using (\ref{Trrhon_CFT}) and (\ref{Mn from Sn}), for the leading term we find
\be
\label{Mn LeadingCFT}
 M^{(n)}_A (b_n)
=
\left(
\frac{\ln (\ell / \epsilon)}{3} 
\right)^n
+
O\Big(\big(\ln (\ell / \epsilon)\big)^{n-1}\Big)\,,
\ee
where the subleading terms in $\ell / \epsilon$ depend both on the non-universal constants and on the parameter $b_n$. 
For instance, in the special case of $n=2$ we get
\be
 M^{(2)}_A (b_2)=\left(
\frac{\ln (\ell / \epsilon)}{3} 
\right)^2
+
\frac{1}{3}\left(
1
+
2 b_2
-2 c'_1
\right)
\ln (\ell / \epsilon)
+O(1)\,,
\ee
where the subleading terms that we have neglected are finite as $\epsilon$ vanishes. 
In the following, with a slight abuse of notation,  
we denote by $\ell$ both the number of consecutive sites in a block $A$ and the length of the corresponding interval $A$ in the continuum. 
This convention is adopted also for the number of sites of a finite chain 
and for the finite size of the corresponding system in the continuum limit, both denoted by $L$.

\subsection{Excited states in CFT}

%Here we will compute both the Entanglement entropy and the capacity of entanglement for some excited states in CFT.  

%$\bullet$
Consider a CFT in a circle of length $L$ in the excited state of the form
$| \textrm{ex}\rangle = \mathcal{O}(0,0) | \textrm{gs}\rangle $,
obtained by applying the operator $\mathcal{O}$ on the ground state. 
The subsystem is an interval $A$ of length $\ell < L$ in a circle of length $L > \ell$.
\\
The R\'enyi entropies in the low-lying excited states in CFT
have been studied in \cite{Alcaraz:2011tn, Berganza:2011mh},
finding that the following ratio provides a UV finite scaling function
\be
\label{Fratio}
F_{\mathcal{O}}^{(n)}(\ell/L)
\equiv
\frac{\textrm{Tr}(\rho_{\mathcal{O}  ,A}^n)}{\textrm{Tr}(\rho_{\textrm{\tiny gs},A}^n)}
=
e^{(1-n)\big( S_{\mathcal{O}  ,A}^{(n)}-S_{\textrm{\tiny gs},A}^{(n)}\big)}\,,
\ee
where $S_{\mathcal{O},A}^{(n)} $ and $S_{\textrm{\tiny gs},A}^{(n)} $ 
denote the R\'enyi entropies when the system is either in the excited state 
or in the ground state respectively.
The moments $\textrm{Tr}(\rho_{\textrm{\tiny gs},A}^n)$ in (\ref{Fratio})
are (\ref{Trrhon_CFT}) with $W_A$ given by the second expression in (\ref{W_A CFT-known}),
while the R\'enyi entropies for the excited state read
\be
\label{Renyi-exc}
S_{\mathcal{O},A}^{(n)}
=
\frac{c}{6}\left(1+\frac{1}{n}\right) \ln\!\left(\frac{L}{\pi \epsilon}\sin\frac{\pi \ell}{L}\right)
+\frac{1}{1-n}\ln\big[F_{\mathcal{O}}^{(n)}(\ell/L)\big]
+ \ln c_n\,.
\ee
In \cite{Alcaraz:2011tn, Berganza:2011mh} 
it has been found that the ratio (\ref{Fratio}) is obtained from
a proper $2n$-point correlator of $\mathcal{O}$.
This gives
%are the Renyi entropies for an excited state labeled by $\Gamma$ and for the ground state respectively and $\hat{\rho}_{A,\Gamma} $ and $ \hat{\rho}_{A,0}$ are the corresponding reduced density matrices.
%Note that from \eqref{Fratio} we can write the entropy for an excited state (denoted by the index $\Gamma$) in the following form \cite{Berganza:2011mh}
\be
\label{entropy-cap-exc}
S_{\mathcal{O},A}
=
S_{\textrm{\tiny gs},A}
-
\frac{d}{d n}\left(\ln F_{\mathcal{O}}^{(n)}(\ell/L)\right)\!  \Big|_{n=1}\,,
\;\;\qquad\;\;
C_{\mathcal{O},A}
=
C_{\textrm{\tiny gs},A}
+
\frac{d^2}{d n^2}\left(\ln F_{\mathcal{O}}^{(n)}(\ell/L)\right)\!  \Big|_{n=1}\,.
\ee
%\noindent
%$\bullet$
In the following we explicitly consider only two examples of excited states
where \cite{Berganza:2011mh,Essler_2013, Calabrese_2014}
\be
\label{F_O-gamma}
F^{(n)}_{\mathcal{O} } (\ell / L )
=
\big[ f_n(\ell / L )\big]^{\gamma}\,,
\ee
where only the exponent $\gamma$ distinguishes the two states and
\be
\label{F-function-def}
f_n(\ell / L )
\equiv
\left(\frac{2}{n}\sin(\pi \ell / L) \right)^{2n}\,
\left(
\frac{
\Gamma \big(\frac{1}{2} \big[1+n+n\csc( \pi \ell /L) \big] \big)
}{
\Gamma \big(\frac{1}{2} \big[1-n+n\csc( \pi \ell /L) \big] \big)
}
\right)^2\,.
\ee
%\noindent
%$\bullet$
From (\ref{entropy-cap-exc}), (\ref{F_O-gamma}) and (\ref{capacityequalentropy}),
one obtains the following UV finite combination
\bea
\label{S-minus-C-current}
S_{\mathcal{O}  ,A} &-& C_{\mathcal{O} ,A}
=
-\,\gamma
\Big(
\partial_n \big[\ln f_n(\ell/L)\big]  \big|_{n=1}
+
\partial^2_n \big[\ln f_n(\ell/L)\big]  \big|_{n=1}
\Big)
- c_1'- [\partial^2_n(\ln c_n)]\big|_{n=1}
\nonumber
\\
\hspace{-1.0cm}
\rule{0pt}{.9cm}
&=&
- \,2\,\gamma\left(
\ln\big|2\sin (\pi \ell/L)\big|
+\psi \! \left(\frac{1}{2 \sin (\pi \ell/L)} \right)
+\sin (\pi \ell/L)
\right)
- c_1'
\\
\rule{0pt}{.7cm}
& &
-\, 2\,\gamma\left(
-1
+\frac{1}{ \sin (\pi \ell/L)} 
+\psi' \! \left(\frac{1}{2 \sin (\pi \ell/L)} \right)
-\big[1+ \sin (\pi \ell/L) \big]^2
\right)
- [\partial^2_n(\ln c_n)]\big|_{n=1}\,,
\nonumber
\eea
where $\psi(x)$ is the digamma function and $\psi ' (x)$ its derivative. 
When $\ell /L \ll 1$, 
since $\psi(1/x)\simeq -\ln x$ and $\psi'(1/x)\simeq x$ as $x\to 0$,
we have that $\tfrac{c}{3} \ln(\ell /\epsilon)$ is the leading term 
of both $S_{\mathcal{O} ,A} $ and  $C_{\mathcal{O} ,A} $.
Furthermore,  the combination in (\ref{S-minus-C-current}) becomes $- [\partial^2_n(\ln c_n)]\big|_{n=1}- c_1'$ in this limit. 
%
%\textcolor{red}{\bf [I would remove the following eqs, placing them in the text if needed]}
%\be
%\label{EE-Cap-J-smallloverL}
%S_J=\frac{c}{3}\ln\left(\frac{\ell}{\epsilon}\right)-\ln c_1'+O(\ell/L)
%\;\;\qquad\;\;
%C_J=\frac{c}{3}\log\left(\frac{\ell}{\epsilon}\right)+\log c_1''+O(\ell/L)
%\ee
%\textcolor{red}{\bf [why in this limit the dependence on the excited state disappear?]}
The difference between the excited state and the ground state becomes invisible in the short interval limit.
%$\bullet$
We remark also that the combinations 
$S_{\mathcal{O},A} - S_{\textrm{\tiny gs},A}$ 
and $C_{\mathcal{O},A} - C_{\textrm{\tiny gs},A}$  (see (\ref{entropy-cap-exc})) are UV finite.

%Hereafter in this subsection we will omit the subindex $A$ for the reduced density matrices for sake of simplicity but we will have in mind that $\hat{\rho}_\Gamma$ and $\hat{\rho}_0$ are the reduced density matrices on the region $A$ of the excited and vacuum states respectively.
%\\
%Similarly, starting from \eqref{defcapacity} and \eqref{Fratio}, the capacity for an excited state labeled by $\Gamma$ is
%\be
%C_{A,\Gamma}=C_{A,0}+\left(\frac{d^2}{d n^2}\left(\log F_\Gamma^{(n)}(x)\right)\right)_{\big|n=1},\label{cpacityexc}
%\ee
%where $C_{A,0}$ is the capacity computed in the vacuum state.
%%Here we will omit the subindex $A$ for the R\'enyi entropies for sake of simplicity but we will have in mind that $\rho_\Gamma$ and $\rho_0$ are the reduced density matrices on the region $A$ of the excited and vacuum states respectively. 
%In the following we will omit for simplicity the index $A$ for the reduced density matrices, for the entanglement entropies and for the capacity of entanglement having in mind that all these quantities are computed for both excited states and the ground state for the same subsystem.

\subsection{Massless compact boson}
\label{subsec:compactboson}
%\begin{center}
%\textcolor{blue}{\bf ==== Massless compact boson ====}
%\end{center}

%\noindent
%$\bullet$ 
Our first example CFT is the massless compactified scalar field, whose central charge is $c=1$. Its action is
\begin{equation}\label{compactifiedboson}
I = \frac{g}{4\pi}\int d^2 x \,\partial_\mu \phi \,\partial^\mu \phi,
\end{equation}
with a field compactification radius $R$ such that $\phi\sim \phi + 2\pi j  R$, $j\in\mathbb{Z}$.
Interestingly, 
for this model
it has been found in \cite{Berganza:2011mh} that,
when the excited state is given by a vertex operator
$\mathcal{O} =\, : \! e^{\textrm{i}\alpha \phi+\textrm{i}\bar\alpha \bar\phi}\! :$,
%whose conformal weights are $(\alpha^2/2,\bar\alpha^2/2)$,
the scaling function (\ref{Fratio}) is equal to one identically. Therefore this excited state has the curious property that its bipartite entanglement structure is unchanged from the ground state.
We will thus move to consider other excited states.
%\\
%\textcolor{red}{\bf [any physical observation from this result about the capacity?]}
%\\

%\noindent
%$\bullet$
A non-trivial result for  (\ref{Fratio}) is found
when $\mathcal{O}=\textrm{i} \partial\phi$ is the current.
In this case (\ref{F_O-gamma}) holds with $\gamma=1$
\cite{Berganza:2011mh,Essler_2013, Calabrese_2014};
hence $F_{\textrm{i} \partial\phi }^{(n)}= f_n $.
Although $F_{\textrm{i} \partial\phi }^{(n)}$ is independent of $R$, for a numerical check we consider a specific value of the compactification radius in order to give an explicit value to the non-universal constants in (\ref{S-minus-C-current}). At the self-dual point, namely when $g R^2=1 $, the masseless compact boson can be studied as the continuum of a free fermion on the lattice described by the Hamiltonian
\be
\label{MasslessFFLattice}
\widehat{H}=-\sum_{j=-\infty}^\infty(\hat{c}^{\dagger}_{j+1}\hat{c}_j+\hat{c}^{\dagger}_j \hat{c}_{j+1})\,,
\ee
where the fermionic operators $\hat{c}_j$ satisfy the anti-commutation relations $\{ \hat{c}_j,\hat{c}^{\dagger}_k \}=\delta_{jk}$.
Indeed, the continuum limit of this free fermionic chain is the massless Dirac field theory, which in the low-energy regime is formulated, through bosonization techniques, as a massless compact free boson \cite{DiFrancescobook,Tsvelikbook}.
%\\
%\textcolor{red}{\bf [mention why and give refs]}
%\\
%\textcolor{teal}{Indeed, 
%For a specific value of the compactification radius $\eta=gR^2= 1/2$ (in the notation of \cite{Calabrese:2009ez}), the massless compact boson is mapped through bosonisation techniques  to a free massless Dirac fermion, which, in turn, can be seen as the continuum limit of a free fermionic chain.
The XX spin chain can be mapped into the free fermionic chain by a Jordan-Wigner transformation.
%}
This implies that we can employ 
the non-universal constant term $c_n$ 
found in \cite{JinKorepin04} through the Fisher-Hartwig theorem,
finding that $-c_1'\simeq 0.726 $ and $[\partial^2_n(\log c_n)]\big|_{n=1}\simeq 0.535 $,
as discussed in the Appendix \ref{app:FHstory}.

\begin{figure}[t!]
\vspace{.2cm}
\hspace{-1.1cm}
\centering
\includegraphics[width=1.05\textwidth]{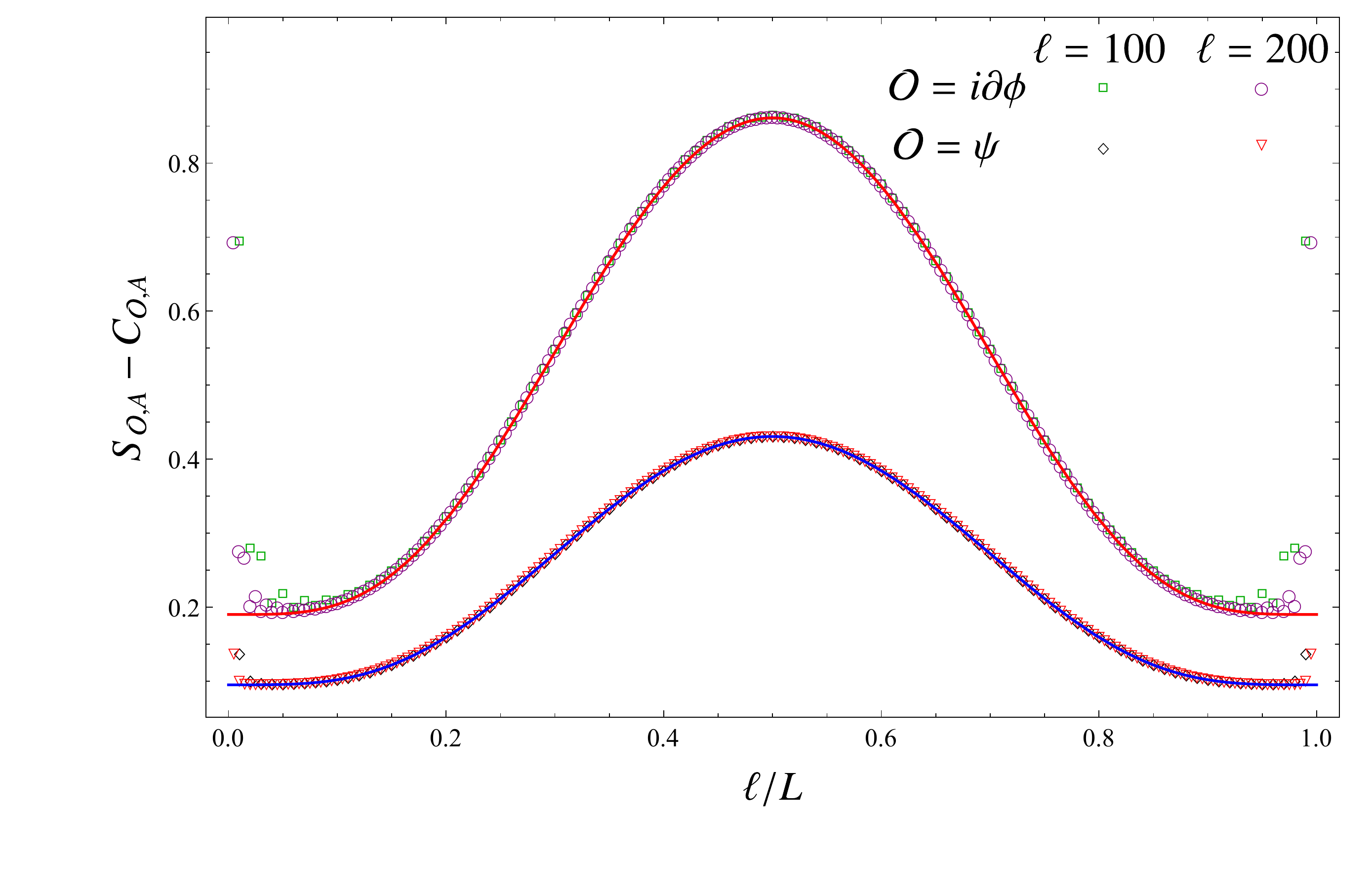}
%\\
%\\
\vspace{-.4cm}
%\hspace{2.8cm}
%\includegraphics[width=.48\textwidth]{excitedpsi}
\caption{
%\textcolor{teal}{
Difference $S_{\mathcal{O}  ,A} - C_{\mathcal{O} ,A}$ as function of $\ell/L$. The data in the top curve have been obtained by considering the excited state of the XX chain corresponding to $\mathcal{O}=\textrm{i} \partial\phi$ in the continuum limit, while the data in the bottom curve have been obtained for the excited state of the Ising chain corresponding to $\mathcal{O}=\psi$ in the continuum limit.
The red and the blue curves are obtained from (\ref{S-minus-C-current}) with $\gamma=1$ and $\gamma=1/2$ respectively. The additive constants for the two curves are reported in Sec.\,\ref{subsec:compactboson} and in Sec.\,\ref{subsec:freefermion}. 
}%}
\label{fig:bosonexcited}
\end{figure}
%%%%%%%%%%%%%%%%%%%%%%%%%%%%%%%%%%%%%%%%%%%%%%%%%%%%%%%%%%%%%%%%%%%%%%%%%%%%%%

%\noindent
%$\bullet$
%In the left panel of Fig.\,\ref{fig:bosonexcited} \dots
%
%\textcolor{red}{\bf [comment the figure for the boson]}
%\\
%\textcolor{teal}{
%In the left panel of 
In Fig.\,\ref{fig:bosonexcited} we show $S_{\textrm{i} \partial\phi,A }-C_{\textrm{i} \partial\phi,A }$ for a block of $\ell$ consecutive sites in periodic chains of free fermions made by $L$ sites.
The top curve is for the XX chain which corresponds to the compact boson CFT. 
The numerical data are obtained through the methods described in \cite{Alcaraz:2011tn, Berganza:2011mh}. 
In the figure, it overlaps with the solid curve, obtained from \eqref{S-minus-C-current} with $\gamma=1$, where the additive constants are specified above. A very good agreement is observed between the numerical data and the CFT predictions for the compact boson. 

\subsection{Free fermion}
\label{subsec:freefermion}
%\begin{center}
%\textcolor{blue}{\bf ==== Free fermion ====}
%\end{center}

%\noindent
%$\bullet$
Let us consider the CFT given by free massless fermion (or, equivalently, the Ising CFT) whose central charge is $c=\frac12$.
In this model, we study the excited states corresponding to the operators $\mathcal{E}$ and $\psi$ \cite{DiFrancescobook}.
%and $\bar \psi$ (whose have conformal weights $(1/2,0)$ and $(0,1/2)$ respectively)
%and $\mathcal{E}$.
\\
In \cite{Berganza:2011mh} it has been shown that, for these states, 
the ratios (\ref{Fratio}) are given by 
(\ref{F_O-gamma}) with $\gamma=1$ and $\gamma=1/2$ respectively, namely
\be
\label{Fratiofermion}
F_\mathcal{E}^{(n)}=F_{\textrm{i} \partial\phi }^{(n)}\,,
\;\;\;\qquad\;\;\;
F_\psi^{(n)}=\sqrt{F_\mathcal{E}^{(n)}}\,.
\ee 
%\noindent
%$\bullet$
%\textcolor{red}{\bf [describe the lattice model and tell why is different from the previous one]}
%\\
%\textcolor{teal}{
The discretization of the Ising CFT is provided by the critical Ising spin chain, whose Hamiltonian can be obtained as a particular case of the one of the XY spin chain. It reads
\be
H_{\textrm{\tiny XY}}= -\sum_{i=1}^L\left(\frac{1+\alpha}{4}\sigma_i^x\sigma_{i+1}^x+\frac{1-\alpha}{4}\sigma_i^y\sigma_{i+1}^y+\frac{\lambda}{2}\sigma_i^z\right)\,,\label{xytext}
\ee
in terms of the Pauli matrices $\sigma_i^{x,y,z}$.
The Hamiltonian of the critical Ising chain and of the XX chain correspond to (\ref{xytext}) with $\alpha=\lambda=1$ and $\alpha=\lambda=0$ respectively.
 Performing a Jordan-Wigner transformation, $H_{\textrm{\tiny XY}}$ is mapped into a chain of free fermions; hence the critical Ising chain and the XX chain (and therefore the numerical data of the two panels of Fig.\,\ref{fig:bosonexcited}) correspond to two different free fermionic models. 
%}

\begin{figure}[t!]
\vspace{.2cm}
\hspace{-1.1cm}
\centering
\includegraphics[width=1.05\textwidth]{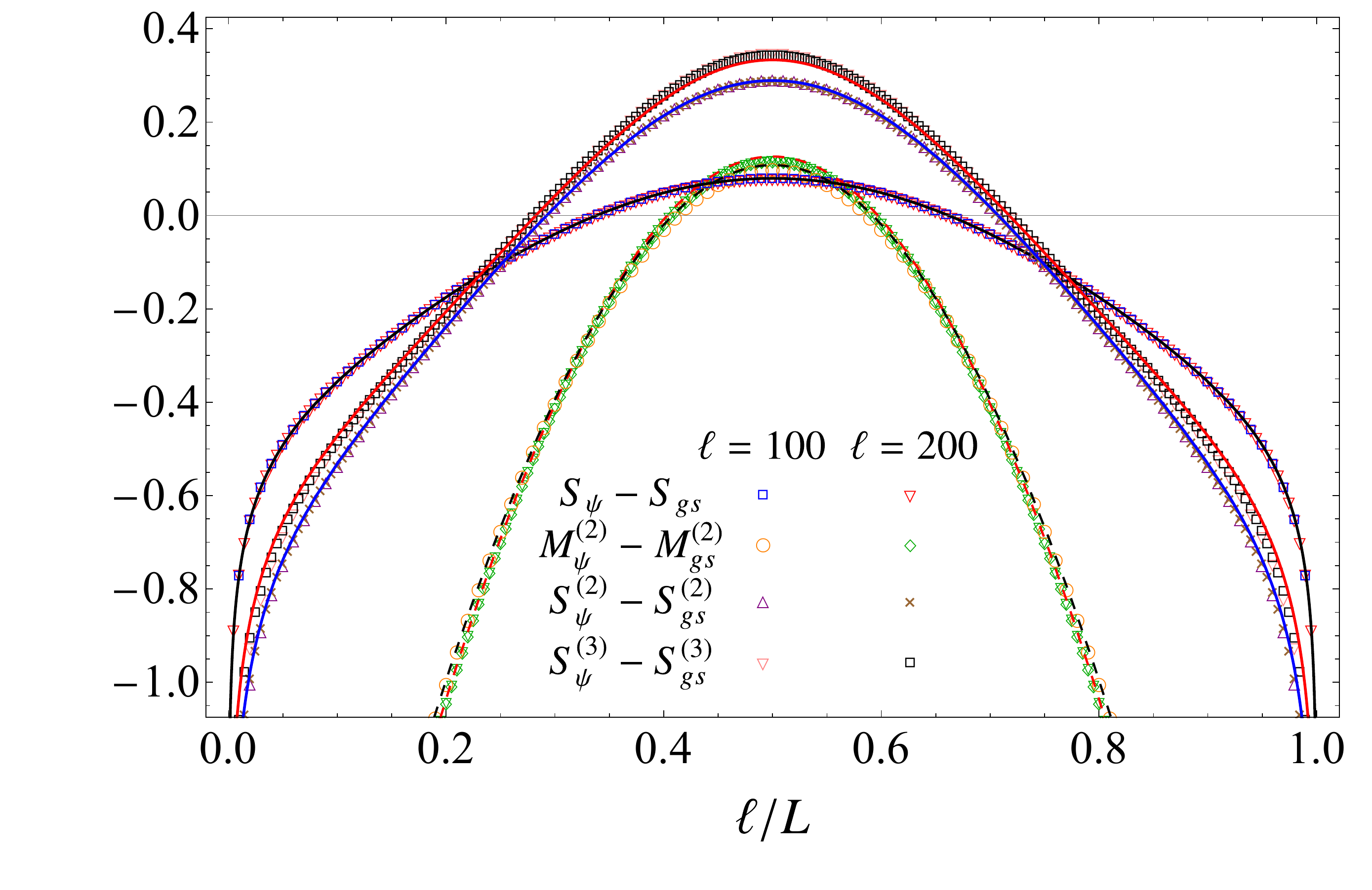}
%\\
%\\
\vspace{-.3cm}
%\hspace{2.8cm}
%\includegraphics[width=.48\textwidth]{excitedpsi}
\caption{
%\textcolor{teal}{
Changes in the monotones as function of $\ell/L$ for moving from the ground state to the excited state of the Ising chain corresponding to $\mathcal{O}=\psi$ in the continuum limit.
The black, the blue and the red solid curves are obtained from (\ref{Renyi-exc}), (\ref{entropy-cap-exc}), (\ref{F_O-gamma}) and (\ref{F-function-def}), while the dashed curves are obtained using (\ref{capacityequalentropy}) and (\ref{entropy-cap-exc}), (\ref{F_O-gamma}) and (\ref{F-function-def}) into (\ref{Mdefinition}) with $\ell=100$ (black dashed curve) and $\ell=200$ (red dashed curve). The non universal constants have been fixed as specified in in Sec.\,\ref{subsec:compactboson} and in Sec.\,\ref{subsec:freefermion}.}
%\textcolor{red}{The outer purple curve corresponds to the change in entanglement entropy $\Delta S_A =S_{A,\psi}-S_{A,gs}$ while the
%inner green curve corresponds to the change in the second moment of modular Hamiltonian $\Delta M_A = M_{A,\psi}-M_{A,gs}$.
%}
\label{fig:deltaS-deltaM}
\end{figure}

%\noindent
%$\bullet$
%In the right panel of Fig.\,\ref{fig:bosonexcited} \dots
%\\
%\textcolor{red}{\bf [comment the panel for the fermion, referring to the appendix, eventually]}
%\\
%\textcolor{teal}{
In Fig.\,\ref{fig:bosonexcited}, the bottom curve is   
$S_{\psi,A }-C_{\psi,A }$ for a block of $\ell$ consecutive sites in a periodic Ising chain made by $L$ spins. The data points are obtained through the procedure detailed in \cite{Alcaraz:2011tn, Berganza:2011mh} and compared with (\ref{S-minus-C-current}) evaluated for $\gamma=1/2$ (blue curve), finding a very good agreement.
%} 
%\\
%\textcolor{teal}{
In this case the additive constants have been fitted, finding $-c_1'\simeq 0.479$ (as already found in \cite{latorre2003ground,Cardy:2007mb}) and $[\partial^2_n(\log c_n)]\big|_{n=1}\simeq 0.385$.

\subsection{Constraints for majorization from monotones}
\label{subsec:thoughtexpCFT}

Let us consider the results in the light of our exploratory approach to majorization  
%of the two thought experiments 
outlined in Sec.\,\ref{subsec:thougthexp}. As anticipated, we focus on the periodic Ising chain, which maps into a chain of free fermions after a Jordan-Wigner transformation, with a finite dimensional Hilbert space. We probe possible majorization between the ground state and an excited state. The excited state corresponds
to the state created by the operator $\mathcal{O} = \psi$ in the continuum CFT. To see if majorization is ruled out, we consider first the changes in the entanglement monotones $S$, $S^{(2)}$, $S^{(3)}$ and $M^{(2)}(\cdot,1)$ for the interval $A$ of length $\ell$. 
The changes in entanglement entropy $\Delta S_A = S_{A,\psi}-S_{A,gs}$, in R\'enyi entropies $\Delta S^{(2)}_A = S^{(2)}_{A,\psi}-S^{(2)}_{A,gs}$ and $\Delta S^{(3)}_A = S^{(3)}_{A,\psi}-S^{(3)}_{A,gs}$ and in the second moment of modular Hamiltonian $\Delta M_{2,A} = M^{(2)}_{A,\psi}(\cdot,1)-M^{(2)}_{A,gs}(\cdot,1)$ increase
monotonically from $\ell/L=0$ to a maximum value when the subsystem is half the size of the chain, $\ell /L=0.5$, as shown in Fig.\,\ref{fig:deltaS-deltaM}. 
%For our fictitious LOCC process
%for state conversion, 
Assuming that we want to convert the ground state to an excited state by a process (such as $\ket{gs} \underset{\textrm{\tiny LOCC}}{\longrightarrow} \ket{\psi}$) 
that involves the order $\ket{\psi}\succ \ket{gs}$,
the entanglement monotones cannot increase. Notice that $\Delta S_A,\, \Delta S^{(2)}_A,\, \Delta S^{(3)}_A$ and $\Delta M_{2,A}$ start as negative for small subsystem size, but all of
them become positive as the size increases, thus ruling out $\ket{\psi}\succ \ket{gs}$ and the LOCC process. Interestingly, $\Delta M_{2,A}$ becomes positive only at $\ell /L\approx 0.403 $, $\Delta S_A$ at $\ell /L\approx 0.337 $, $\Delta S^{(2)}_A$ at $\ell /L\approx 0.292 $ and $\Delta S^{(3)}_A$ at $\ell /L\approx 0.282 $. Thus, in this case the monotone $S^{(3)}_A$ gives the stronger constraint, ruling out $\ket{\psi}\succ \ket{gs}$ in the range $\ell /L \in [0.282,0.718]$. In the opposite transition from the excited state to the ground state with $\ket{gs}\succ \ket{\psi}$, the signs are reversed with $\Delta M_{2,A}\geq 0$ giving a stronger constraint for ruling out
the $\ket{gs}\succ \ket{\psi}$ in the regime $\ell /L \in [0,0.403] \cup [0.597,1]$.

The general picture is consistent with the naive expectation that it becomes relatively harder to connect two pure states by means of LOCC as the difference between the sizes of the two subsystems $A$ and $B$ becomes smaller. A heuristic argument is that the dimension of the space of allowed local operations decreases as $|d_A-d_B|$ gets smaller. If one e.g. looks at unitaries, the total dimension of the group of local unitaries $U(d_A)\times U(d_B)$ is $d_A^2+d_B^2$ which is minimized for $d_A=d_B$. Beyond this, it seems hard to extract general lessons from our preliminary analysis.

\section{Summary and outlook}
\label{sec:conclusions}

{\bf Partial orders among quantum states.} The sequences of inequalites coming from the monotones introduced in Sec.\,\ref{sec:resource_monotones} can be used to define a partial order among quantum states.  For example, starting with the sequence $M^{(n)} (\rho;n-1)$, we could define $\rho\succ_n\sigma$
if all monotones up to degree $n$ obey $M^{(k)}(\rho;n-1) \leq M^{(k)}(\sigma;n-1)$, and then define $\rho \succ_\infty \sigma$ by $\lim_{ k\rightarrow\infty} \rho \succ_k \sigma$. With the extremal sequence $P^{(n)}_{E} (\rho)$
we have even more freedom, due to the infinite range of parameters $\vec{a}=(a_1,\ldots ,a_{\lceil\frac{n}{2}\rceil-1})$, where $\lceil x\rceil$ is the smallest integer greater than or equal to $x$. As we did for the order $n=3,4$ cases, we could first express $P^{(n)}_E$ in terms of $M^{(k)}$ up to order $n$, and then find the 
value of parameters $\vec{a}_0$ that produces the tightest inequality. The extremal parameter vector $\vec{a}_0$ can be expressed as a combination of $M^{(k)}$, leading to a non-linear expression for $P^{(n)}_{E,\vec{a}_0}$ in terms of 
$M^{(k)}$. We can then define a tighter partial order $\rho \succ^E_n \sigma$ by requiring  $P^{(k)}_{E,\vec{a}_0}(\rho) \geq P^{(k)}_{E,\vec{a}_0}(\sigma)$ for all $k\leq n$, and a limit $\rho \succ^E_\infty \sigma$ as in the above. 

It is an  interesting question whether any of the above partial orders is equivalent to the majorization partial order. It is easy to see in examples that $\rho \succ^E_n \sigma$ with finite $n$ does not imply majorization, even in low dimensional systems\footnote{Consider e.g. a three-dimensional Hilbert space and let's look at the first two extremal monotones, $-S$ and $P^{(2)}_E$. Suppose $\rho$ has eigenvalues $0.49$, $0.41$ and $0.1$, and that $\sigma$ has eigenvalues
$0.5$, $0.3$ and $0.2$. Then $\rho$ does not majorize sigma, but the difference between $-S$ is $\sim 0.084$ and between 
$P^{(2)}_E$ is $\sim 0.256$, which are both positive.}. It is also known that monotonicity of R\'enyi entropies is not sufficient to imply majorization. All in all, our best hope is
perhaps that the order $\succ^E_\infty$ is strong enough to imply majorization. The hope is based on the equivalence (the Hardy-Littlewood-P\'olya inequality of majorization \cite{HLPbook})
\bea\label{majconvex}
&&\rho \succ \sigma\quad \Leftrightarrow \quad \Tr f(\rho) \geq \Tr f(\rho) \\
&& \textrm{for all real valued continuous convex functions}\  f \ \textrm{defined on} \ [0,1]. \nonumber 
\eea
It is known that
convex functions can be expanded {\em e.g.} as a series of Bernstein polynomials. One could hope that a series expansion based on the extremal sequence $P^{(n)}_{E,\vec{a}_0}$ would be possible with all coefficients being non-negative, or in other words, that non-negative linear combinations of the extremal polynomials are dense in the space of real values continuous functions. 
Then $\succ^E_\infty$ would imply that the  inequality on the right hand side of (\ref{majconvex}) is satisfied, and be equivalent to majorization. The above mentioned partial order based on monotones and majorization partial order can also be formulated as orderings generated by cones, this concept is discussed {\em e.g.} in \cite{MaOlAr}. We present this reformulation
in Appendix \ref{app:coneappendix}.

Another interesting open question is a version of a moment problem.  For a state $\rho$ in a $d$-dimensional system, it is known that $d-1$ first R\'enyi entropies $S^{(k)}(\rho)$, $k=1,\ldots ,d-1,$ are sufficient to determine the spectrum of $\rho$.
The explicit steps and comments on the history of this observation can be found in \cite{boes2020}. 
The proof is actually straightforward, as the R\'enyi entropies yield a basis for the symmetric polynomials in the eigenvalues, which can then directly be used to compute ${\rm det}(\lambda-\rho)$ whose roots are the spectrum of $\rho$.
Suppose that one knows all $M^{(k)}$ or all $P_E^{(k)}$ up to order $n$, for a state $\rho$ in a $d$-dimensional system. Is it possible to derive the spectrum of $\rho$ for
some value of $n$ or in the limit $n\rightarrow \infty$? If not, can even a partial spectrum be calculated? 

{\bf Inequalities for R\'enyi entropies.} In this paper, we considered convex 
functions of the type ${\rm Tr}[\rho F(\log\rho) ]$,
which have the feature that expressions of this type include
entanglement entropy and moments of shifted modular Hamiltonian.
We could also have considered even simpler functions
of the type ${\rm Tr}[F(\rho)]$ with polynomials $F$ which are 
essentially linear combinations of R\'enyi entropies
with integer powers. The function $F$ is convex when $F''\geq 0$, and 
we can once again find a complete basis of
extremal polynomials. In this case, we need to find positive 
polynomials on the interval $[0,1]$ and these are given
by linear combinations (with non-negative coefficients) of polynomials 
of the form $\prod_i (x-a_i)^2$ or $x(1-x) \prod_i
(x-a_i)^2$ with in either case $a_i\in [0,1]$. Notice that linear 
combinations of such polynomials can yield a polynomial of a lower 
degree, and one therefore has to be a bit careful to find all 
polynomials of a particular degree.
For example, the most general linear $F''$ is a non-negative linear 
combination of $x$ and $1-x$, and since
$x=x(1-x) + x^2$ and $1-x= x(1-x)+ (1-x)^2$ these can indeed both be 
written as linear combinations of
the extremal basis polynomials of higher degree.
If $F''=x$ then the monotone is simply ${\rm Tr}[\rho^3]$, and for 
$F''=1-x$ we obtain the new monotone
$\frac{1}{2} {\rm Tr}[\rho^2] -\frac{1}{6} {\rm Tr}[\rho^3]$. Going to 
higher degrees, one could in principle
obtain infinite families of monotones. This family would then 
presumably be complete, in other words, imposing all of them would be
equivalent to state majorization. It would be interesting to study 
this in more detail.

{\bf Inequalities for quantum field theories.} As we discussed, to define
 majorization in quantum field theory directly
 requires one to introduce an explicit UV cutoff. It is however not 
 obvious that this is a natural construction
 as the notion of majorization may depend sensitively on the choice of 
 UV cutoff. Since relative entropy, as opposed
 to entanglement entropy, is well defined for continuum quantum field 
 theories, it is tempting to think that only
 a relative version of majorization applies in continuum quantum field 
theories. This leads one to consider
 the inequality $S(\rho_1||\sigma) \geq S(\rho_2 ||\sigma)$ in quantum 
 field theory. This inequality would follow
 if there exists a quantum channel ${\cal N}$
%, i.e. a linear completely 
 %positive trace preserving map (CPTP), 
which maps
 $\rho_1$ to $\rho_2$ and maps $\sigma$ to itself. For general quantum 
 channels monotonicity of relative entropy is the
 statement that $S(\rho||\sigma) \geq S({\cal N}(\rho) || {\cal 
 N}(\sigma) )$. Similar monotonicity properties are satisfied by R\'enyi relative entropies (R\'enyi divergences). In \cite{Bernamonti:2018vmw} monotonicity
constraints $S_\alpha (\rho (0)||\gamma_\beta) \geq S_\alpha (\rho (t)||\gamma_\beta)$ were investigated as additional "second laws" constraining the 
off-equilibrium dynamical evolution $\rho (t) = {\cal N}_t (\rho (0))$ (where the Gaussian state is a fixed point $\gamma_\beta = {\cal N}_t (\gamma_\beta)$) in 2d CFTs and their gravity duals.

 To  speculate, one could try an 
 alternative method to construct inequalities and proceed as follows.
 In quantum field theory, the definition
 of relative entropy also requires a choice of algebra, typically 
 associated to a subregion. If we can replace the
 action of the channel on states by the adjoint action ${\cal N}^{\ast}$ 
 on the algebra $\mathcal{A}$, defined via ${\rm Tr}({\cal N}(\rho) O)={\rm Tr}(\rho 
 {\cal N}^{\ast}(O))$ for all $\rho$ and $O\in \mathcal{A}$, then we can also write 
 $S_{\cal A}(\rho||\sigma) \geq
 S_{{\cal N}^{\ast} {\cal A}} (\rho || \sigma )$. This inequality 
 follows from a corresponding operator inequality
 for the relative modular operators, $\Delta_{\rho|\sigma;{\cal A}} \leq 
 \Delta_{\rho|\sigma;{\cal N}^{\ast}{\cal A}}$.
 We could take this inequality to be the fundamental inequality which 
 defines a quantum field theory counterpart
 of majorization. It would define a partial ordering for algebras 
 (given two states), rather than for states.
 By applying operator monotones\footnote{A complete classification of 
 operator monotones is known, besides the linear
 function $f(x)=x$ all other operator monotones are non-negative 
 (possibly infinite) linear combinations of functions of the
 type $f(x)=x/(x+s)$ with some $s>0$ \cite{loewner}.} we could then derive additional 
 inequalities in the spirit of the paper. We leave a further 
 exploration of these ideas to future work.

{\bf Additional open questions \cite{Wilming_private_communication}.}
It is worth asking whether it is possible to generalize from \cite{boes2020} the Result 1 (a sufficient
condition for approximate state transition) or the Result 2 (bounds on smoothed min and max entropies)
to involve higher cumulants than entropy and variance. As far as we can see, 
these results rely on the Cantelli-Chebyshev inequality for deviations of a random variable from its mean value,
with the bound depending on the variance. One could try to employ a refined inequality involving higher cumulants as well,
and then try to construct extensions of the abovementioned results.  Finally, it would be interesting to explore if our approach to resource monotones has interesting applications in other quantum resource theories. In particular, 
it would be interesting to study resource monotones in the context of (un)complexity and its connections to quantum gravity.

\section*{Acknowledgments}

We thank Henrik Wilming for useful comments and questions. 
GDG, EKV and ET acknowledge Galileo Galilei Institute for warm hospitality and financial support 
(through the program {\it Reconstructing the Gravitational Hologram with Quantum Information})
during part of this work.  JdB is supported by the European Research Council under
the European Union's Seventh Framework Programme (FP7/2007-2013), ERC Grant
agreement ADG 834878. 
EKV's research has been conducted within the framework of InstituteQ - the Finnish Quantum Institute, and  ET's research within the framework of the Trieste Institute for Theoretical Quantum Technology (TQT).

\appendix

\section{Details on the construction of entanglement monotones}
\label{App:DetailsSec2}

In this appendix, we provide additional discussions on the construction of entanglement monotones for pure states detailed in Sec.\,\ref{subsec:concavity}.
In particular, in Appendix \ref{subsec:highercumulants} we describe a general procedure for obtaining entanglement monotones from the cumulants of the modular Hamiltonian, while in Appendix \ref{app:Proof} we report a detailed proof of the Theorem \ref{thm_Pol}. 

\subsection{Higher cumulants and pure state entanglement monotones}
\label{subsec:highercumulants}

In this subsection we point out that there are many ways to construct generalizations of the function $M$ defined in (\ref{Mdefinition}) to concave quantities (entanglement monotones) 
\be
\label{property Mn}
 M^{(n)}(\boldsymbol{\gamma})=\textrm{Tr}\big[\rho F_n( \rho )\big]\,,
\ee
involving higher cumulants or moments of modular Hamiltonian, where $x\,F_n(x)$ is a concave function of a single real variable $x\in [0,1]$. 
With a slight abuse of notation, we denote the quantity in (\ref{property Mn}) like the one in (\ref{Mn Tr Fn}) 
although the former one is more general 
because of the occurrence of the parameters $\boldsymbol{\gamma}=(\gamma_0,\dots,\gamma_{n-1} )$, as discussed below.

%\subsubsection{Moments \& Cumulants}

Let us list explicitly the relation (\ref{momentstocumulants}) of moments $\mu_n$ and the cumulants $C_n$ for first moments: 
\bea
\label{M1-from-C}
\mu_1 & = & C_1\,,
\\
\mu_2 & = & C_1^2 + C_2\,,
\\
\mu_3 & = & C_1^3 +3\, C_1\, C_2 + C_3\,,
\\
\mu_4 & = & C_1^4 +6\, C_1^2\, C_2 +3\, C_2^2 +4\, C_1\, C_3 + C_4 \,.
\label{M4-from-C}
\eea
One can invert to obtain the relation of cumulants to moments, for example
\bea
C_1 & = & \mu_1\,,
\\
C_2 & = & \mu_2 - \mu_1^2\,,
\\
\label{C3-from-M}
C_3 & = & \mu_3 - 3\, \mu_2\, \mu_1 + 2\, \mu_1^3\,,
\\
\label{C4-from-M}
C_4 & = & \mu_4 - 4\, \mu_3\, \mu_1 + 12\, \mu_2 \, \mu_1^2 - 6\, \mu_1^4 - 3\, \mu_2^2 \,.
\eea

In our case, given a density matrix $\rho$, we are interested in the moments of modular Hamiltonian
$\mu_n=\textrm{Tr}\big(\rho (-\ln \rho )^n\big)=\textrm{Tr}\big( \rho K^n\big)$, where $K=-\ln \rho $ and likewise for the cumulants. The entropy is
$C_1=\mu_1=S$ and the capacity is $C_2=C$.

Instead of considering moments of shifted modular Hamiltonian $\Tr [\rho (-\ln \rho +b_n)^n]$ with $b_n>n-1$ as concave generalizations of $M$, there are more general constructions.
Here is one way to proceed.  For $n=1$ let us consider
\be
M^{(1)}
\equiv  C_1 + a_1
=
\mu_1 + a_1\,.
\ee
Up to an additive constant, $M^{(1)}$ is the entropy. 
%\\
%\textcolor{red}{\bf $\bullet$ [what is the criterion to fix $a_1$?]}
%In order to employ Casini's argument, $M_1 >0$.
%\\
%
%\noindent
%{\bf Second cumulant}
%
For $n=2$ let us consider
\be
\label{f2-step1}
M^{(2)} 
\equiv  C_2 + \big( \mu_1 + a_1 \big)^2 + a_2
=\,
\mu_2 + 2a_1 \mu_1 + a_1^2 + a_2
=\,
\mu_2 + \gamma_1 \mu_1 + \gamma_0\,,
\ee
which depends on the two parameters $\gamma_0$ and $\gamma_1$.
Comparing (\ref{f2-step1}) with (\ref{Mdefinition}), we find that, when $\gamma_1=2$ and $\gamma_0=1$, $M^{(2)}(\gamma_0,\gamma_1)$ reduces to $M$ in (\ref{Mdefinition}). Notice that the polynomial combination of capacity and entropy in $M^{(2)}$ reduces to a linear combination of moments $\mu_n$ of modular Hamiltonian.
%\\
%\textcolor{red}{\bf $\bullet$ [what is the criterion to fix $a_1$ and $a_2$?]}
%\\
%
%\noindent
%{\bf Third cumulant}

Likewise, for $n=3$, if we start with a polynomial expression
\be
\label{f3-step1}
M^{(3)}
\,\equiv\,
C_3 
+ \frac{3}{2} \big( \mu_2 + \mu_1 + b_1 \big)^2
-2 \big( \mu_1 + c_1 \big)^3
+ \alpha_2 \big(\mu_2 + a_2\big)^2
+ \alpha_1 \big(\mu_1 + a_1\big)^2 \ ,
\ee
with the coefficients $3/2$ and $-2$ in (\ref{f3-step1})
we can cancel the terms $\mu_2 \mu_1$ and $\mu_1^3$ respectively 
(see (\ref{C3-from-M})).
Imposing the vanishing of the coefficients of $\mu_2^2$ and $\mu_1^2$ in (\ref{f3-step1})
leads to
\be
\alpha_2 = - \frac{3}{2}\,,
\;\;\qquad\;\;
\alpha_1 = - \frac{3}{2} + 6\, c_1\,.
\ee
The point is that we are again lead to a linear combination of moments,
\be
M^{(3)}(\gamma_0,\gamma_1,\gamma_2)
\,=\,
\mu_3 
+ \gamma_2\,  \mu_2 
+ \gamma_1\, \mu_1 
+ \gamma_0\,,
\ee
with three parameters
\be
\gamma_2 \equiv 3 \big(b_1 - a_2\big) \,,
\qquad
\gamma_1 \equiv 3 \big(b_1 -2 c_1^2 + 4 c_1 a_1 - a_1\big)\,,
\qquad
\gamma_0 \equiv -\frac{1}{2}  \big( 4 c_1^3 - 3 b_1^2 - 12 c_1 a_1^2 + 3a_2^2 + 3 a_1^2 \big) \ .
\ee
%Th parameters can be chosen such that .
%
%\textcolor{red}{\bf $\bullet$ [what is the criterion to fix these four parameters?]}
%\%\
Thus, since $\textrm{Tr} \rho = 1$, we have 
\be\label{M3}
M^{(3)} (\gamma_0,\gamma_1,\gamma_2)
\,=\,
\textrm{Tr}\big[ \rho \, F_3(\rho )\big] + \gamma_0\,,
\;\;\qquad\;\;
F_3(\rho ) \equiv
- \,( \ln \rho  )^3 
+ \gamma_2 \,( \ln \rho )^2
- \gamma_1 \ln \rho\,.
\ee
It is easy to see that there is a range of parameters $\gamma_1,\gamma_2$ such that $f_3(x) =x\, F_3(x)$ is concave in the unit interval. The difference to the third
moment of shifted modular Hamiltonian 
$\textrm{Tr}[\rho (-\ln \rho + b_3)^3]$ is that it contains a single parameter $b_3$. It is a special case of (\ref{M3}) with
\be
  \gamma_j =(-1)^j{3 \choose 3-j}b^{3-j}_3, \ \ \textrm{where}\ j=0,1,2.
\ee
We are thus lead to consider more general linear combinations of moments as an alternative generalization of the measure $M$ (of which the moments of shifted modular Hamiltonian (\ref{lincombmoments}) are a special case):
\bea
\label{fn-step1}
M^{(n)}(\boldsymbol{\gamma})&\equiv & \mu_n+\sum_{j=1}^{n-1} \gamma_j \mu_j+\gamma_0
\\
&=&
\textrm{Tr}\bigg[ \rho \bigg((-1)^n\ln^n \rho+\sum_{j=1}^{n-1} \gamma_j (-1)^j\ln^j \rho \bigg)\bigg] + \gamma_0 = \textrm{Tr}[\rho F_n (\rho)] +\gamma_0\,,
\eea
where $\boldsymbol{\gamma}=\big(\gamma_0,\dots,\gamma_{n-1}\big)$. The range of the parameters $\gamma_{j\neq 0}$ can be chosen so that $x F_n(x)$ is a concave function for $x\in [0,1]$ and $\gamma_0$ so that $M^{(n)}(\boldsymbol{\gamma})\geq 0$. It is
also clear that these measures can be computed by using the R\'enyi entropies or $\textrm{Tr}(\rho^\alpha)$ as a generating function, by applying a combination of derivatives
$\sum_k \gamma_k (-1)^k\, \partial^k_\alpha$ and setting $\alpha =1$. 

To summarize, there are many ways to construct infinite sequences of entanglement monotones, generalizing $M$, and compute them from R\'enyi entropies. In the end, the desiderable
construction depends on the specific physical motivation.

\subsection{Proof of Theorem \ref{thm_Pol}}\label{app:Proof}

According to Theorem \ref{thm_Pol},  all positive semidefinite polynomials $G(y)$ on the negative half-line $y\in (-\infty ,0]$ have the following form. 
For polynomials $G(y)$ of degree $2d$ (with $d\geq 1$) they are linear combinations with positive coefficients of polynomials of the form 
$G_{\vec{a}}(y)= \prod^d_{i=1} (y+a_i)^2$, with all $a_i\geq 0$. For polynomials of degree $2d+1$ they are linear combinations with positive coefficients of polynomials of the form
$G_{\vec{a}}(y)=-y~\prod^d_{i=1} (y+a_i)^2$, with again all $a_i\geq 0$.

\begin{proof} Consider first a positive polynomial on the entire real line. 
It can be written as $\prod_i (x-x_i)$ where the roots can be complex. There cannot be an isolated real root, as then the
polynomial would be negative somewhere in a small neighborhood of that real root. Similarly, there can not be an odd degeneracy
of a real root, because once more the function would be negative in a small neighborhood. Therefore all real roots need
to have even degeneracy. So the polynomial is of the form $q(x)^2 r(x)$ where $q(x)$ is real and all other (complex) roots
make up $r(x)$. Because the polynomial must be real, the roots must come in complex conjugate pairs. Therefore $r(x)=|s(x)|^2$
where $s(x)$ contains all the roots in (say) the complex upper half plane. We can write $s=s_0+is_1$ where $s_0$ and $s_1$
are the real and imaginary parts. Then we see that the polynomial is of the form $q(x)^2 s_0(x)^2 + q(x)^2 s_1(x)^2$ which shows
that a positive polynomial on the real line must be sum of two squares.

Now consider a polynomial $p(x)$ which is positive on the negative real axis. We can decompose these polynomials
again in roots. Negative real roots need to appear with even multiplicity and positive real roots can appear
with any multiplicity. Factors of the type $|x-u|^2$ with complex $u$ are positive definite and can appear
without restriction. Consider now polynomials of the form $f-xg$ with $f$ and $g$ positive on the entire real
axis. These polynomials form a ring (so if you multiply two it will still be of this form). The claim is that $p$
is also of this form (which is manifestly non-negative on the negative real axis). We simply need to check that
all factors of $p$ are of this form. A factor with negative real roots with even multiplicity is of the form
$f+x\cdot 0$ as it is positive on the entire real axis. A factor with a positive real root can be written as
$u-x$ which is also of the required form (with $f=u$ and $g=1$). Finally, factors $|x-u|^2$ are positive definite
on the entire real axis and therefore also of the form $f+x\cdot 0$. Using the previous characterization of 
positive polynomials we conclude that $p(x)$ can be written as
\be 
p(x)=q(x)^2 + r(x)^2 - x [s(x)^2 +t(x)^2]\,,
\ee
for some polynomials $q,r,s,t$. This result is due to P\'olya-Szeg\"o \cite{PolSze}.
It remains to show that each of these terms can be written as a liner combination of extremal polynomials.
Look e.g. at $q(x)^2$ and expand it in the form
\be
q(x)^2=\prod_{i,j} (x-u_i)^2 [(x-b_j)^2 +c_j^2]^2\,.
\ee
This is a sum of terms of the form $a_-(x)^2 a_+(x)^2$ with positive coefficients, where $a_-(x)$ has zeroes
on the negative real axis, and $a_+(x)$ has zeroes on the positive real axis. We can further expand $a_+(x)^2$
as a power series with alternating coefficients. This shows that $q(x)^2$ is indeed a linear combination of extremal
polynomials with non-negative coefficients. The same result applies for the other three terms in $p(x)$. This completes the proof.
\end{proof}

\section{Capacity of entanglement in fermionic chains: constant term}
\label{app:FHstory}

In this appendix we exploit the method of \cite{JinKorepin04} to determine the non-universal constant occurring 
in the expression of the capacity of entanglement for a block $A$ made by $\ell$ consecutive sites 
in the infinite free fermionic chain. 
%These corrections to the logarithmic leading behaviour are not accessible using field theoretical techniques in the continuum limit.

The Hamiltonian of the free fermionic chain on the line reads
\be 
\label{Ham FF generich}
H=-\,\sum_{n=-\infty}^{+\infty}
\!  \bigg[\,
\hat{c}^\dag_n \,\hat{c}_{n+1}+\hat{c}^{\dag}_{n+1}\,\hat{c}_n
-
2 h \bigg(\hat{c}^\dag_n \,\hat{c}_{n}-\frac{1}{2}\bigg)\bigg]\,,
\ee
where 
$\{\hat{c}_n^\dag,\hat{c}_m^\dag\}=\{\hat{c}_n,\hat{c}_m\}=0$ and $\{\hat{c}_n,\hat{c}^\dag_m\}=\delta_{m,n}$ 
and $h$ is the chemical potential. The ground state of this model is a Fermi sea with a Fermi momentum $k_{\textrm{\tiny F}}=\arccos|h|$. 
A Jordan-Wigner transformation maps the Hamiltonian (\ref{Ham FF generich}) into the Hamiltonian of the XX spin chain with magnetic field $h$.

The Toeplitz nature of the correlation matrix restricted to $A$ for this model 
allows us to write the large $\ell$ expansion of $\ln\textrm{Tr}\rho_A^n$ 
through the Fisher-Hartwig conjecture, which provides the asymptotic behaviour of the Fredholm determinant of Toeplitz matrices for large matrix size. 
The result reads \cite{JinKorepin04}
\be
\label{logZn_FHcorr}
\ln\textrm{Tr}\rho_A^n
\,=\,
\frac{1}{6}\left(
\frac{1}{n}-n\right)\ln\ell+\ln c_n
\,=\,
\frac{1}{6}\left(
\frac{1}{n}-n\right)\big[\ln \ell +\ln (2|\sin(k_{\textrm{\tiny F}})|)\big]
+ \Upsilon(n)
+ o(1)\,,
\;\;
\ee
where
\be 
\label{Upsiloncorr}
\Upsilon(n)= \mathrm{i} n\int_{-\infty}^\infty 
\big[ \tanh(\pi w)-\tanh(\pi n w) \big]
\ln\left(
\frac{\Gamma\big(\frac{1}{2}+\mathrm{i}w \big)}{\Gamma\big(\frac{1}{2}-\mathrm{i}w \big)}
\right)\,.
\ee
By introducing $G_n(w)\equiv n [\tanh(\pi w)-\tanh(\pi n w)]$, we need
\be
\label{der1and2ofFwrtn}
\partial_n  G_n(w)\big|_{n=1}
=
-\frac{\pi w}{\cosh^2(\pi w)}\,,
\;\;\qquad\;\;
\partial^2_n  G_n(w)\big|_{n=1}
=
-\frac{2\pi w}{\cosh^2(\pi w)}+\frac{2\pi^2 w^2}{\cosh^2(\pi w)}\tanh(\pi w)\,.
\ee
Plugging (\ref{der1and2ofFwrtn}) into the derivatives of (\ref{Upsiloncorr}), one finds  
the corrections to the entanglement entropy and the capacity of entanglement due to $\Upsilon(n)$ in (\ref{logZn_FHcorr}). 
This gives
\be
\label{SA_FHcorr}
S_A
=
\frac{1}{3}\ln \ell + \frac{1}{3}\ln (2|\sin(k_{\textrm{\tiny F}})|)  -\Upsilon'(1) +\dots\,,
\ee
which has been obtained in \cite{JinKorepin04}, and 
\be
\label{CE_FHcorr}
C_A
=
\frac{1}{3}\ln \ell + \frac{1}{3}\ln (2|\sin(k_{\textrm{\tiny F}})|) + \Upsilon''(1) +\dots\,,
\ee
where the constant $\Upsilon'(1)$ and $\Upsilon''(1)$
can be evaluated numerically 
from (\ref{der1and2ofFwrtn}) in (\ref{Upsiloncorr}), 
finding $-\Upsilon'(1)\simeq 0.495018$ and $\Upsilon''(1)\simeq 0.303516$.
The subleading terms that we have neglected are vanishing as $\ell \to \infty$ 
and some of them have been computed in \cite{CalabreseEssler_10_XXchain} through the generalised Fisher-Hartwig conjecture.
%We show a very good agreement between the numerical data and the formulas (\ref{SA_FHcorr}) and (\ref{CE_FHcorr}) in Sec.\,\ref{sec:num_FF_eq}.

In the main text we have mainly considered 
(\ref{SA_FHcorr}) and (\ref{CE_FHcorr}) in the case of  vanishing chemical potential, i.e. for $h=0$, which means $k_{\textrm{\tiny F}}= \frac{\pi}{2}$
(see e.g. all the figures in Sec.\,\ref{CFT}).

\section{Sequences of monotones and an ordering generated by a cone}
\label{app:coneappendix}

In the Sec.\,\ref{sec:conclusions} we discussed partial orders among quantum states based on the sequences of our new monotones. Here we rephrase this question in terms of an ordering generated by a cone \cite{MaOlAr}. By diagonalizing a density matrix, the space of quantum states in $d+1$ dimensions can be identified with
the standard simplex $\Delta^d \subset \mathbb{R}^{d+1}$. Our monotones can be thought as convex functions
\be
M: \Delta^d \rightarrow \mathbb{R}, \ M(\xv ) \equiv \sum^d_{i=0} x_i F(\ln x_i )\,,
\ee
such that 
\be
G(y)\equiv F''(y)+F'(y)\,,
\ee
with $y\equiv \ln x$ is a non-negative polynomial of the order $n-1$ on the negative half-line $(-\infty ,0]$. Thus $F$ is polynomial of degree $n$.
The above functions form a convex cone $\Ccal_n$. For our purposes we may identify functions that differ by a constant. Every function $M$ is a linear combination with positive
coefficients of the extremal rays of the cone. Let $\Phi_n$ denote the set of extremal rays. We say that the set $\Phi_n$ {\em generates} the cone $\Ccal_n$.
 For $n=1$ there is only one extremal ray with $F(y) = y$, thus $\Phi_1=\{\sum_i x_i \ln x_i\}$. For $n>1$ we found in Theorem 1 that 
the extremal rays correspond to functions $F_{\vec{a}}(y)$ with $G_{\vec{a}}(y)=F_{\vec{a}}''(y)+F_{\vec{a}}(y)$ of the form
\be
  G_{\vec{a}}(y) \equiv \left\{ \begin{array}{l} \prod^{k}_{i=1} (y+a_i)^2 \, ,\; \,\,\,\,\,\,\,\, a_i \geq 0\ \ \forall i    \,, \ {\rm when}\ n-1=2k \geq 2\,, \\                                                                  -y \prod^k_{i=1} (y+a_i)^2 \,,\; a_i \geq 0\ \ \forall i  \,, \ {\rm when}\ n-1=2k+1 \geq 3\,.  \end{array} \right.
\ee
Thus in both cases the generating set $\Phi_n$ is infinite, parameterized by vectors ${\vec{a}}$ in the hyperorthant of $\mathbb{R}^k$, so the cone $\Ccal_n$ is infinitely generated. Finally 
let us define the cone 
\be
        \Ccal = \cup^\infty_{n=1} \Ccal_n\,,
\ee
which is infinitely generated by the set
\be
     \Phi = \cup^\infty_{n=1} \Phi_n \, .
\ee

Majorization $\rho \succ \sigma$ in $d+1$ dimensions can be identified with majorization of vectors ${\bs \lambda} \succ {\bs \mu}$ or  majorization partial order in the 
standard simplex $\Delta^d$. Now alternatively \cite{MaOlAr} we can define an ordering $\succ_{\Ccal}$ based on the function set $\Phi$ that generates the cone $\Ccal$:
\be
  {\bs \lambda} \succ_{\Ccal} {\bs \mu}  \Leftrightarrow M({\bs \lambda}) \geq M({{\bs \mu}})  \ \forall M\in \Phi \, .
\ee
The inequality on the right hand side is satisfied by every function $M\in \Ccal$, the definition just uses the most economical set of functions generating the cone.
The ordering $\succ_\Ccal$ is said to be {\em generated} by the cone $\Ccal$. Such (partial) orderings come with a basic problem. Define the {\em completion} $\Ccal^*$ of $\Ccal$, the cone
of all functions that respect the ordering $\succ_\Ccal$
\be
 \Ccal^* \equiv \left\{ f: \Delta^d \rightarrow R~||~\xv \succ_\Ccal \yv \Rightarrow f(\xv) \geq f(\yv) \right\}\,.
\ee
A basic problem is to {\em identify} the completion $\Ccal^*$ of $\Ccal$, which is an important open question for the cone defined above.

We noted the Hardy-Littlewood-P\'olya inequality of majorization
\bea
  &&\rho \succ \sigma \Leftrightarrow \Tr g(\rho) \geq \Tr g(\sigma) \nonumber \\
&&\textrm{for all continuous convex functions } g:[0,1]\rightarrow \mathbb{R} \ . \nonumber
\eea
We could alternatively interpret this as another ordering generated by a cone. Define the convex cone
\be
\Ccal_{HLP} = \left\{ f: \Delta^d \rightarrow R~||~f(\xv) = \sum^d_{i=1} g(x_i)\ , g:[0,1]\rightarrow \mathbb{R} \textrm{ is continuous, convex} \right\}\,,
\ee
and define the ordering generated by the cone $\Ccal_{HLP}$,
\be
\xv \succ_{\Ccal_{HLP}} \yv \Leftrightarrow f(\xv) \geq f(\yv) \ \forall f\in \Ccal_{HLP}\,.
\ee
Then by the HLP inequality we can identify majorization with the cone ordering,
\be
\rho \succ \sigma \Leftrightarrow {\bs \lambda}\succ_{\Ccal_{HLP}} {\bs \mu} \ .
\ee
Now we can ask if $\Ccal_{HLP}$ could be in the completion $\Ccal^*$ of $\Ccal$ or at least well approximated by $\Ccal^*$. This would mean that $\succ$ and $\succ_\Ccal$ are equivalent. 
We have thus reformulated the question posed in Sec.\,\ref{sec:conclusions} as a problem of comparing orderings generated by cones. 
\bibliographystyle{nb}

\bibliography{bibTex_Capacity_Final-v3}
%\newpage
%%%%%%%%%%%%%%%%%%%%%%%%%%%%%%%%%%%%%%%%%%%%%%
%\begin{thebibliography}{99}
%\ifx\href\asklfhas\newcommand{\href}[2]{#2}\fi
%\ifx\arxivref\asklfhas\newcommand{\arxivref}[2]{\href{http://arxiv.org/abs/#1}{#2}}\fi
%\ifx\doiref\asklfhas\newcommand{\doiref}[2]{\href{http://dx.doi.org/#1}{#2}}\fi
%\raggedright
%\small
%\parskip 0pt
%%
%%
%
%
%
%\end{thebibliography}

\end{document}